\documentclass{article}

\usepackage{arxiv}

\usepackage[utf8]{inputenc} 
\usepackage[T1]{fontenc}    
\usepackage{hyperref}       
\usepackage{url}            
\usepackage{booktabs}       
\usepackage{amsfonts}       
\usepackage{nicefrac}       
\usepackage{microtype}      
\usepackage{lipsum}
\usepackage{graphicx}
\graphicspath{ {./images/} }

\usepackage{bm,amsmath,bbm,amsfonts,amsthm,amsbsy,amscd,amsxtra,amsgen,amsopn,amssymb,mathtools,caption}
\usepackage{threeparttable}
\usepackage[usenames,dvipsnames]{xcolor}
\usepackage{natbib}

\newcommand{\bs}[1]{\boldsymbol{#1}}
\newtheorem{theorem}{Theorem}
\newtheorem{lemma}{Lemma}
\newtheorem{conditions}{Conditions}

\title{Joint modelling of longitudinal and time-to-event data applied to group sequential clinical trials}

\author{Abigail J. Burdon$^{1,*}$, 
Lisa V. Hampson$^{2}$ and Christopher Jennison$^{3}$ \\
$^{1}$MRC Biostatistics Unit, University of Cambridge, Robinson Way, Cambridge, CB2 0SR, U.K\\
$^{2}$Statistical Methodology, Novartis Pharma AG, Basel BA2 7AY, Swizerland\\
$^{3}$Department of Mathematical Sciences, University of Bath, Claverton Down, Bath BA2 7AY, U.K \\
$^{*}$\texttt{Email: abigail.burdon@mrc-bsu.cam.ac.uk}}

\begin{document}
\maketitle
\begin{abstract}
Often in Phase 3 clinical trials measuring a long-term time-to-event endpoint, such as overall survival or progression-free survival, investigators also collect repeated measures on biomarkers which may be predictive of the primary endpoint. Although these data may not be leveraged directly to support early stopping decisions, can we make greater use of these data to increase efficiency and improve interim decision making? We present a joint model for longitudinal and time-to-event data and a method which establishes the distribution of successive estimates of parameters in the joint model across interim analyses. With this in place, we can use the estimates to define both efficacy and futility stopping rules. Using simulation, we evaluate the benefits of incorporating biomarker information and the affects on interim decision making. 
\end{abstract}


\section{Introduction}
\label{sec:intro}
Interest in joint modelling is motivated by clinical trials where the biomarker is predictive of a time-to-event (TTE) outcome. Overall survival and progression-free survival are examples of TTE endpoints. For example, \citet{goldman1996response} use CD4 lymphocyte cell counts as a surrogate endpoint for survival in a clinical trial comparing the efficacy and safety of two antiviral drugs for HIV-infected patients. \citet{taylor2013real} use joint models to predict survival times of patients with prostate cancer based on prostate specific antigen (PSA) levels measured by blood tests at multiple hospital visits.

To date, most group sequential designs for survival trials focus on a single primary endpoint. Designs which leverage repeated measurements on continuous or binary endpoints have been reported such as those by \citet{galbraith2003interim} and this work extends this idea to survival trials with repeated measurements on a predictive biomarker. We shall focus on monitoring randomised control trials so that patients are randomised at baseline to receive either a novel treatment or control and the objective is to test for superioirity of the treatment versus control. Hence, let \(\eta\) denote the treatment effect in a statistical model and let \(\eta_0\) be the true value of \(\eta\), then we shall test the null hypothesis
\begin{equation}
\label{eq:null}H_0:\eta_0\leq 0 \hspace{5mm} \text{vs}\hspace{5mm} H_A:\eta_0 > 0.
\end{equation}
In the group sequential setting with a total of \(K\) analyses, we shall test \(H_0\) at analys \(k\) by defining treatment effect estimates \(\hat{\theta}^{(k)},\) information levels \(\mathcal{I}_k\) and standardised \(Z-\)statistics given by \(-\hat{\eta}^{(k)}\sqrt{\mathcal{I}_k}\) for \(k=1,\dots,K.\) The aim will then be to determine the upper boudary constants \(b_1,\dots,b_K\) and lower boundary constants \(a_1,\dots,a_K\) for the group sequential test as in Figure~\ref{fig:boundaries}. These boundaries are calculated to ensure such that the Type 1 error rate  is \(\alpha\) and power is \(1-\beta\) when \(\eta=-0.5\). In particular, we shall use the error spending test, also known as \citet{gordon1983discrete}, to calculate the boundary points.

\begin{figure}[t]
\centering
\includegraphics[width=0.9\textwidth]{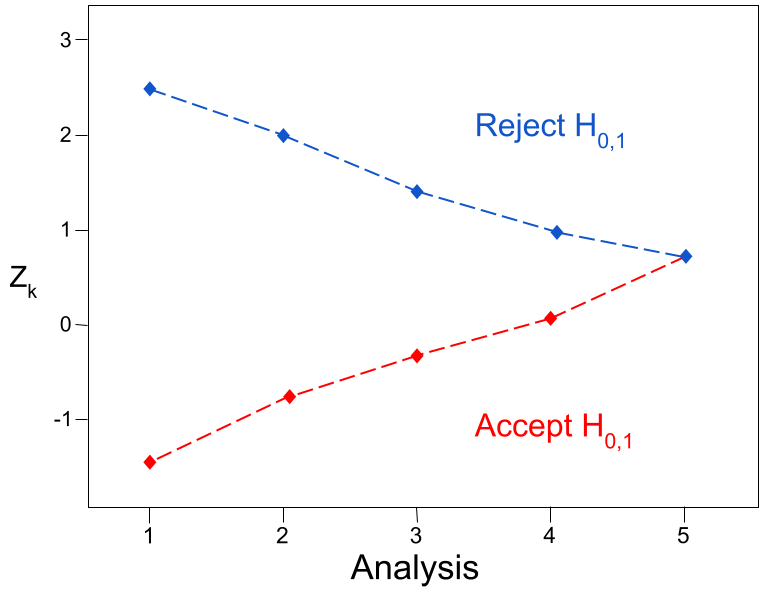}
\caption{Group sequential trial with \(k=5\) analyses. Blue points represent the upper boundary constants \(b_1,\dots, b_K\) and red points give the lower boundary constants \(a_1,\dots,a_K.\).}
\label{fig:boundaries}
\end{figure}

We shall develop methods for designing and analysing a group sequential trial based on a joint model for longitudinal and TTE data.  We may believe that a trend in the trajectory of the biomarker is predictive of the TTE outcome, and we would like to know whether these additional longitudinal data can be used to inform early stopping decisions. Suppose that biomarker observations are available but have not been used in the analysis. We shall assess the change in efficiency of the trial when these observations are included in the analysis. We shall focus on efficiency measured in terms of the number of patients that need be recruited to achieve a certain power, and we and show that, in some scenarios, the trial using the longitudinal data is up to 1.7 times as efficient as the trial which discards the longitudinal data.

\section{Joint modelling}\label{sec:joint}
\subsection{Joint model}
\label{subsec:joint_model}
The joint model that we consider is given in Equation (2) by \cite{tsiatis2001semiparametric}. We shall refer to the authors as TD for short. There are two processes in this model which represent the survival and longitudinal parts separately, and these processes are linked through the hazard rate of the survival process. First we consider the longitudinal data. Suppose that \(X_i(t)\) is the true value of the biomarker at time \(t\) for subject \(i\) and that \(W_i(t)\) is the observed value of the biomarker at time \(t\) for patient \(i\). Then the longitudinal model takes the form
\begin{align}
\label{eq:long1}
X_i(t)&=b_{0i}+b_{1i}t \\
\label{eq:long2}
W_i(t)&= X_i(t)+\epsilon_i(t)
\end{align}
where \(\mathbf{b}_i=(b_{i0},b_{i1})\) is a vector of patient specific random effects and \(\epsilon_i(t)\) is the measurement error. In general, the vector \(\mathbf{b}_i\) can have dimension \(p\) and the function \(X_i(t)\) need not be constrained to a linear function of \(t.\) We consider a random effects model where each \(\mathbf{b}_i\) is a random variable with density function \(f(\cdot).\) The measurement errors are assumed to be independent and if the biomarker for patient \(i\) is measured at times \(t_{i1},\dots t_{im_i}\), then \(\epsilon_i(t_{ij})|\mathbf{b}_i\sim N(0,\sigma^2) \text{ for } j=1,\dots,m_i\) and \(\epsilon(t)\) and \(\epsilon(t')\) are independent for \(t\neq t'.\) For now, we shall assume that \(\sigma^2\) is known and we later describe how \(\sigma^2\) can be estimated.

The model for the survival endpoint is a Cox proportional hazards model in which the underlying trajectory \(X_i(t)\) acts as a time-varying covariate with coefficient \(\gamma\). Let \(Z_i\) be an indicator function that patient \(i\) receives the experimental treatment and let \(\eta\) be the corresponding treatment effect. Then the hazard function is given by
\begin{equation}
\label{eq:surv}
h_i(t)=h_0(t)\exp\{\gamma X_i(t)+\eta Z_i\},
\end{equation}
where \(h_0(\cdot)\) is the baseline hazard function. In general, \(Z_i\) may be a \(p\times 1\) column vector of coefficients and  \(\eta\) is the corresponding coefficient vector of length \(p\). In this model, \(\gamma\) is the parameter that determines the correlation between the longitudinal data and the TTE endpoint. We show in Section~\ref{sec:results} that if \(\gamma=0\) then the investigator may pay a small penalty for having tried to leverage the biomarker data. However, if \(\gamma>0\), then there is a large benefit from fitting this joint model to the data. Together, Equations~\eqref{eq:long1}, \eqref{eq:long2} and \eqref{eq:surv} define the joint model.

\subsection{Conditional Score}
\label{subsec:conditional_score}
In the fixed sample setting, TD present the ``conditional score" method for fitting the joint model to the data. The method adapts the general theoretical work by \citet{stefanski1987conditional} who find unbiased score functions by conditioning on certain parameter-depedent sufficient statistics. This is a desireable method because the analysis is semi-parametric so that there are no distributional assumptions required for the random effects. Further, the authors show that the parameter estimates are normally distributed and unbiased. We shall now extend the fixed sample theory for the joint model to group sequential trials. To perform a group sequential trial with \(K\) analyses, we need to know the joint distribution of the sequence of treatment effect estimates that will be obtained at analyses \(k=1,\dots,K.\) To determine this distribution, we shall define group sequential versions of all objects included in the single-stage conditional score, which are calculated using data obtained at that analysis. Equivalent definitions of TD's single-stage conditional score, and associated functions, can be found by setting \(K=k=1\)

The conditional score methodology builds upon the theory of counting processes. In the general definition, a counting process is a step-function increasing in integer increments and the survival counting process is a step function jumping from 0 to 1 at the failure time for an uncensored observation.

The censoring mechanism is used to keep patients in the at-risk set who have yet to experience an event. For patient \(i\) with time-to-failure random variable \(F_i\), let \(C_i(k)\) be the time-to-censoring random variable at analysis \(k\). This censoring event includes ``end of study" censoring for the total follow-up time of patient \(i\) at analysis \(k\), then at analysis \(k\) the event time random variable is \(T_i(k)=\min\{F_i,C_i(k)\}.\) The observed event time at analysis \(k\) is \(t_i(k)\) and the observed censoring indicator is \(\delta_i(k)=\mathbb{I}\{F_i\leq C_i(k)\}.\) In the conditional score approach, to be included in the at-risk set at time \(t\) the patient must have at least two longitudinal observations to fit the longitudinal regression model. The at-risk process at analysis \(k\) is an indicator for not yet observing the event, not yet censored, or not having enough longitudinal observations. Let \(v_{i2}\) be the time of the second longitudinal observation for patient \(i\), then at analysis \(k\), the at-risk process and counting process  for the joint model are
\begin{align*}
Y_i(k,t) &=\mathbb{I}\{t_i(k)\geq t, v_{i2}\leq t\} \\
N_i(k,t) &= \mathbb{I}\{t_i(k) \leq t, \delta_i(k) = 1, v_{i2} \leq t\}.
\end{align*}

The function
\begin{align*}
dN_i(k,t) &= N_i(k,t+dt)- N_i(k,t^{-}) \\
&= \mathbb{I}\{t\leq t_i(k) < t+dt, \delta_i(k) = 1, v_{i2} \leq t\}
\end{align*}
presents us with useful notation: for any function or stochastic process \(f(\cdot),\) the stochastic integral
\begin{equation}
\label{eq:stochastic_integral}
\int_0^\infty f(u)dN_i(k,u)=f(t_i)
\end{equation}
is \(f\) evaluated at the place where \(N_i(k,t)\) jumps from 0 to 1 if \(\delta_i(k)=1\) and \(v_{i2} \leq t\), and 0 otherwise.

By analogy to survival analysis, we seek a compensated counting process with expectation zero. This property leads us to define an estimating equation from which we can obtain treatment effect estimates that are asymptotically normally distributed by Section~5.3 by~\citet{van2000asymptotic}. In the usual survival analysis setting, we can calculate the compensated counting process by subtracting the intensity process from the counting process itself. In the joint modelling setting, this  is not as simple because the randomness of the nuisance parameters \(\mathbf{b}\) mean that the intensity process is unpredictable. To overcome this, TD introduce a ``conditional intensity process" which is conditional on a certain ``sufficient statistic". The origins of the conditional intensity process and sufficient statistic are not crucial for our purpose. What we actually use are the definitions and properties that are derived from these. In what follows, we shall present TD's definitions of the single-stage conditional intensity process, sufficient statistic and compensated counting process and we extend these to the group sequential versions. We shall then show that the group sequential compensated counting process has expectation zero.

For patient \(i\), let \(v_i(u)\) be set of all time points for measurements of the biomarker, up to and including time \(u\). Let \(\hat{X}_i(u)\) be the ordinary least squares estimate of \(X_i(u)\) for patient \(i\) based on the set of measurements taken at times \(v_i(u)\). That is, calculate \(\hat{b}_{0i}(u)\) and \(\hat{b}_{1i}(u)\) based on measurements taken at times \(v_i(u)\), then \(\hat{X}_i(u)=\hat{b}_{0i}(u)+\hat{b}_{1i}(u)u\). As we pass time \(v_{ij},\) a new observation \(W_{ij}\) is generated and the formula for \(\hat{X}_i(u)\) is updated for larger values of \(u\). This seems strange since at an early time point, \(s\) where \(s < u\), we use data \(v_i(s)\) in the calculation of \(\hat{X}_i(s)\) even though there may be more data available at time points \(v_i(u)\). However, this is necessary for the martingale property to hold in later results. Suppose that \(\sigma^2\theta_i(u)\) is the variance of the estimator \(\hat{X}_i(u)\) at time \(u\). TD define the sufficient statistic to be the function
\begin{align*}
S_i(k,t,\gamma,\sigma^2)&=\hat{X}_i(t)+\gamma\sigma^2\theta_i(t)dN_i(k,t)\\
&=\hat{b}_{0i}(t)+\hat{b}_{1i}(t)t+ \gamma\sigma^2\theta_i(t)dN_i(k,t)
\end{align*}
which is defined for all \(t\in (v_{i2},t_i)\) for patient \(i\). The corresponding conditional intensity process and compensated counting process are defined by
\begin{align}
\label{eq:cond_intensity_group1}
\lambda^C_i(k,t) &= lim_{dt\downarrow 0}\frac{\mathbb{P}(dN_i(k,t) = 1|S_i(k,t, \gamma,\sigma^2), t_i(t),Z_i,Y_i(k,t))}{dt}\\
\label{eq:cond_intensity_group2}
&=h_0(t)\exp\{\gamma S_i(k,t,\gamma,\sigma^2)-\gamma^2\sigma^2\theta_i(t)/2+\eta^T Z_i\}Y_i(k,t) \\
\notag
&=h_0(t)E_{0i}(t,\gamma,\eta,\sigma^2)Y_i(k,t) \\
\notag
M_i(k,t) &= N_i(k,t)-\int_0^t\lambda^C_i(k,t)du \\
\notag
dM_i(k,t) &= dN_i(k,t)-\lambda^C_i(k,t).
\end{align}
We show below that this compensated counting process has expectation zero conditional on \(S_i(k,t,\gamma,\sigma^2)\) and use this result to obtain the asymptotic distribution of some parameter estimates in the joint model. Specifically, we will determine the distribution of the estimates \(\hat{\gamma}^{(k)}, \hat{\eta}^{(k)}\) and \(\hat{\sigma}^{(k)2}\) which are unbiased estimates, at analysis \(k\), of \(\gamma,\eta\) and \(\sigma^2\) respectively.

\begin{lemma}
\label{lemma:E0_joint}
The function \(dM_i(k,t)\) has expectation zero conditional on the sufficient statistic, that is
\[\mathbb{E}(dM_i(k,t)|S_i(k,t, \gamma,\sigma^2), t_i(t),Z_i,Y_i(k,t))=0.\]
\end{lemma}
\begin{proof} See \citet{tsiatis2001semiparametric}.
\end{proof}

TD present the fixed sample conditional score in their Equation~(6). We present some functions that are needed to define the group sequential conditional score at analysis \(k\) and the derivate of such an object with respect to parameters \(\gamma\) and \(\eta\). Let
\[J_i(k,t, \gamma,\sigma^2) = \{S_i(k,t, \gamma, \sigma^2)-\hat{X}_i(t)-\gamma\sigma^2\theta_i(t), 0,\dots,0\}^T, \hspace{1cm}
\Gamma_i(k,t,\sigma^2) = \begin{bmatrix}
\sigma^2\theta_i(t) & \mathbf{0} \\  \mathbf{0}  &  \mathbf{0}  \end{bmatrix}\]
be a \((p+1)\times 1\) vector and \((p+1)\times (p+1)\) matrix respectively. Now, for the following functions, we drop the dependency of all functions on the parameters \(k, t, \gamma,\eta\) and \(\sigma^2.\) Then the functions of interest are

\begin{align}
\notag
S_c^{(0)} &= \frac{1}{n}\sum_{i=1}^n Y_iE_{0i}
&C^{(1)}&= \frac{1}{n}\sum_{i=1}^n J_iY_iE_{0i} \\
\notag
S_c^{(1)} &= \frac{1}{n}\sum_{i=1}^n \left\{\begin{array}{c}
S_i\\
Z_i
\end{array}\right \}Y_iE_{0i}
&C^{(2)} &= \frac{1}{n}\sum_{i=1}^n \left\{\begin{array}{c}
S_i\\
Z_i
\end{array}\right \}J_i^T Y_iE_{0i} \\
\label{eq:SC}
S_c^{(2)} &= \frac{1}{n}\sum_{i=1}^n \left\{\begin{array}{c}
S_i\\
Z_i
\end{array}\right \} \left\{\begin{array}{c}
S_i\\
Z_i
\end{array}\right \}^TY_iE_{0i}
&C^{(3)}&= \frac{1}{n} \sum_{i=1}^n\Gamma_i(u)dN_i(k,u)Y_i E_{0i} \\
\notag
E_c&= \frac{S_c^{(1)}}{S_c^{(0)}}
&V^{(2)}_c &= \frac{C^{(2)}}{S_c^{(0)}}-\frac{S_c^{(1)}C^{(1)^T} }{[S_c^{(0)}]^2 } \\
\notag
V^{(1)}_c &= \frac{S_c^{(2)}}{S_c^{(0)} }-\frac{S_c^{(1)}S_c^{(1)^T} }{[S_c^{(0)}]^2 }
&V^{(3)}_c &= \frac{C^{(3)}}{S_c^{(0)} }.
\end{align}
These functions are analogous to to those in \citet{jennison1997group} for survival data. In the simple survival model, we have that \(S^{(1)}=\partial S^{(0)}/\partial\bs{\theta} \) and \(S^{(2)}=\partial  S^{(1)}/\partial \bs{\theta}\) where \(\bs{\theta}\) is the vector of parameters in the hazard function of the simple survival model. In a similar manner, the superscripts \((1),(2)\) and \((3)\) refer to the order of the differentiation for certain functions which we explain further in Section~\ref{subsec:diff}. The function \(E_c(k,t,\gamma,\eta,\sigma^2)\) has the interpretation of the expectation of the vector \(\{S_i(k,t,\gamma,\sigma^2),Z_i^T\}^T\) at analysis \(k\) weighted by the conditional intensity process. Let \(\tau_k\) be the maximum follow-up time at analysis \(k\), then the conditional score for analysis \(k\) is
\begin{equation}
\label{eq:cond_score_k}
U_c(k,\gamma,\eta,\sigma^2) = \int_0^{\tau_k}\sum_{i=1}^{n}\left( \{S_i(k,u, \gamma,\sigma^2), Z_i^T\}^T- E_c(k, u, \gamma,\eta, \sigma^2)\right) dN_i(k, u).
\end{equation}

\subsection{Differentiation of the conditional score}
\label{subsec:diff}
Another object of importance is the first derivative of the conditional score function, Equation~\eqref{eq:cond_score_k}, with respect to parameters \(\gamma\) and \(\eta.\) This matrix plays a key role in the definition of the covariance matrix for the estimates \(\hat{\gamma}\) and \(\hat{\eta}\) and has a likeness to the Fisher information matrix which is the derivative of the score statistic for general statistical models. TD comment that the variance matrix can be found, however they do not present an equation for such an object. Details for the following calculation can be found in the supplementary materials. The first derivative the conditional score of Equation~\eqref{eq:cond_score_k} with respect to \(\gamma\) and \(\eta\) is
\begin{align}
\notag
\frac{\partial}{\partial(\gamma,\eta)^T}U_c(k,\gamma,\eta,\sigma^2) = \sum_{i=1}^{n} \int_0^\infty \bigg[&\Gamma_i(k,u,\sigma^2) -V_c^{(1)}(k,u,\gamma,\eta,\sigma^2) \\
\label{eq:d_cond_score}
+  &V_c^{(2)}(k,u,\gamma,\eta,\sigma^2)  + V_c^{(3)}(k,u,\gamma,\eta,\sigma^2)\bigg]dN_i(k,u).
\end{align}

\subsection{Asymptotic theory for parameter estimates \(\hat{\gamma},\hat{\eta}\) and \(\hat{\sigma}^2\) in the joint model}
\label{subsec:asymptotic_theory}

We shall now prove that the estimates \(\hat{\gamma}^{(1)}, \hat{\eta}^{(1)},\dots,\hat{\gamma}^{(K)}, \hat{\eta}^{(K)}\) are asymptotically multivariate normally distributed and we shall derive an explicit form for the covariance matrix of this vector of parameters. In what follows, we shall determine the limiting distribution as \(n\rightarrow\infty.\) For clarity, we denote the objects \(\hat{\gamma}_n^{(k)}, \hat{\eta}_n^{(k)}\) and \(U_c^{(n)}(k, \gamma,\eta,\sigma)\) as dependent on \(n\).

\begin{theorem}
\label{theorem:joint_group}
Suppose that \(\gamma_0,\eta_0\) and \(\sigma^2_0\) are the true values of the parameters \(\gamma,\eta\) and \(\sigma^2\) respectively. For each \(k=1,\dots,K\), let \(\hat{\gamma}_n^{(k)}\) and \(\hat{\eta}_n^{(k)}\) be the values of \(\gamma\) and \(\eta\) which are the solution to the equation \((U_c^{(n)}(k, \gamma,\eta,\sigma_0)= 0\) and suppose that Conditions~\ref{conditions:conditions_joint_group} hold.  Then the vector \((\hat{\gamma}_n^{(1)}, \hat{\eta}_n^{(1)}, \dots, \hat{\gamma}_n^{(K)}, \hat{\eta}_n^{(K)})^T \) converges in distribution to a multivariate Gaussian random variable, specifically
\[n^{\frac{1}{2}}\left(\begin{array}{c}
\hat{\gamma}_n^{(1)} -\gamma_0 \\
\hat{\eta}_n^{(1)} -\eta_0 \\
\vdots \\
\hat{\gamma}_n^{(K)} -\gamma_0 \\
\hat{\eta}_n^{(K)} -\eta_0 \\
\end{array}\right)
\xrightarrow{d} N
\left(\begin{bmatrix}\mathbf{0}\\ \mathbf{0}\\ \vdots\\\mathbf{0}\end{bmatrix}, 
\Sigma=\begin{bmatrix} \Sigma_{11} & \Sigma_{12} & \cdots & \Sigma_{1K} \\ \Sigma_{12} & \Sigma_{22} & \cdots & \Sigma_{2K} \\ \vdots & \vdots & \ddots & \vdots \\ \Sigma_{1K} & \Sigma_{2K} & \cdots & \Sigma_{KK}\end{bmatrix}\right)\]
where 
\begin{equation}
\label{eq:sigma}
\Sigma_{k_1k_2}= (A^{(k_1)})^{-1}B^{(k_1)}((A^{(k_2)})^{-1})^T
\end{equation}
 and the matrices \(A^{(k)}\) and \(B^{(k)}\) are defined by
\begin{align}
\label{eq:group_A}
A^{(k)} &= \int_0^{\tau_k}\left[\mathbb{E}(\Gamma_i(k,u, \sigma^2_0)) -v_c^{(1)}(k,u,\gamma_0,\eta_0,\sigma_0^2)- v_c^{(2)}(k,u,\gamma_0,\eta_0,\sigma_0^2)\right]s^{(0)}_c(k,u,\gamma_0,\eta_0,\sigma_0^2)h_0(u)du  \\
\label{eq:group_B}
B^{(k)} &=\int_0^{\tau_k} v_c^{(1)}(k,u,\gamma_0,\eta_0,\sigma_0^2)s^{(0)}_c(k,u,\gamma_0,\eta_0,\sigma_0^2)h_0(u)du.
\end{align}

\end{theorem}
\begin{proof}
We begin this proof by considering the set of \(K(p+1)\) stacked equations
\begin{equation}
\label{eq:est_eqs_all}
\begin{bmatrix}U_c(1, \gamma,\eta,\sigma^2) \\ \vdots \\ U_c(K,\gamma,\eta,\sigma^2)\end{bmatrix} = \begin{bmatrix} \mathbf{0} \\ \vdots \\ \mathbf{0} \end{bmatrix}
\end{equation}
and we shall show that these equations define a multidimensional estimating equation. That is, for each \(k=1,\dots,K\), the function \(U_c(k,\gamma,\eta,\sigma^2)\) has expectation zero. The arguments given by TD about asymptotic normality in the fixed sample case apply for each \(k=1,\dots,K\). We present some of these arguments which are useful for the derivation of the covariance matrix in the group sequential setting. We shall write \(U_c(k,\gamma,\eta,\sigma^2) \) as 
\begin{align}
\label{eq:cond_group_a}
&\int_0^{\tau_k}\sum_{i=1}^{n}\left( \{S_i(k,u, \gamma,\sigma^2), Z_i^T\}^T- e_c(k,u, \gamma,\eta, \sigma^2)\right) dM_i(k,u) \\
\label{eq:cond_group_b}
+&\int_0^{\tau_k}\sum_{i=1}^{n}\left(e_c(k,u, \gamma,\eta, \sigma^2)- E_c(k,u, \gamma,\eta, \sigma^2)\right) dM_i(k,u).
\end{align}
By the same argument as TD, \(n^{-1}\) times Expression~\eqref{eq:cond_group_b} converges in probability to zero in a neighbourhood of \((\gamma_0,\eta_0)\) and we deduce that the behaviour of the estimates \(\hat{\gamma}^{(k)}\) and \(\hat{\eta}^{(k)}\) will be dictated by Expression~\eqref{eq:cond_group_a}. Then we have that the expectation of Expression~\eqref{eq:cond_group_a} is
\begin{equation}
\label{eq:exp_0}
\mathbb{E}\left( \int_0^{\tau_k}\sum_{i=1}^{n}\left( \{S_i(k,u, \gamma,\sigma^2), Z_i^T\}^T- e_c(k,u, \gamma,\eta, \sigma^2)\right) dM_i(k,u)\right)=0
\end{equation}
for each \(k=1,\dots,K\). Therefore combined with the fact that Expression~\eqref{eq:cond_group_b} converges in probability to zero, we have that Equation~\eqref{eq:est_eqs_all} defines a multidimensional estimating equation. Therefore, the estimates \((\hat{\gamma}^{(k)},\hat{\eta}^{(k)})\) for each \(k=1,\dots,K\) are asymptotically multivariate normal and unbiased for parameters \(\gamma\) and \(\eta\) by Section~5.3 by~\citet{van2000asymptotic}. It remains to determine the covariance matrix.

In the general setting for statistical models, the sandwich estimator can be used to robustly estimate the variance matrix for estimates that are the solutions to estimating equations as in Section~2.6 by~\citet{wakefield2013bayesian}. The sandwich estimator requires deriving \(B = Var(U(\theta))\) and \(A = \partial U(\theta)/\partial\theta \) where \(U(\cdot)\) is an estimating function for parameter \(\theta.\) The similarity here is with matrices \(A^{(k)}\) and \(B^{(k)}\) in Equations~\eqref{eq:group_A} and~\eqref{eq:group_B} which are the sequential versions of such objects for the conditional score. \citet{wakefield2013bayesian} prove the asymptotic normality of estimates obtained using estimating equations and the asymptotic convergence of the sandwich estimator. The proof by \citet{wakefield2013bayesian} applies to scalar estimates and one-dimensional estimating functions so we shall extend this proof to derive the covariance matrix for the vector of parameters in the joint model as oppose to the scalar sandwich variance.

Following \citet{wakefield2013bayesian}, we apply the standard Taylor expansion results to each row of Equation~\eqref{eq:est_eqs_all} and aggregate the results. Details of this calculation are found in the supplementary materials. Let \(\bar{\mathbf{A}}\) be the block diagonal matrix whose \(k^{th}\) diagonal matrix is the \(p\times p\) matrix \(\mathbf{A}^{(k)}.\) Then we obtain
\[-n^{-\frac{1}{2}}\bar{\mathbf{A}}^{-1}\begin{bmatrix}U_c(1, \gamma,\eta,\sigma^2) \\ \vdots \\ U_c(K, \gamma,\eta,\sigma^2) \end{bmatrix}  = \bar{\mathbf{A}}^{-1}\begin{bmatrix} n^{-1} \frac{\partial}{\partial(\gamma,\eta)^T} U_c(1,\gamma,\eta, \sigma^2)|_{(\gamma^{*(1)}\eta^{*(1)})} \\ \vdots \\n^{-1} \frac{\partial}{\partial(\gamma,\eta)^T} U_c(K,\gamma,\eta, \sigma^2)|_{(\gamma^{*(K)}\eta^{*(K)})}\end{bmatrix}  \cdot n^{\frac{1}{2}}\begin{bmatrix}
\hat{\gamma}_n^{(1)} -\gamma_0 \\
\hat{\eta}_n^{(1)} -\eta_0 \\
\vdots \\
\hat{\gamma}_n^{(K)} -\gamma_0 \\
\hat{\eta}_n^{(K)} -\eta_0 \\
\end{bmatrix}\]
where \((\gamma^{*(k)}\eta^{*(k)})\) lies on the line segment between \((\gamma_0,\eta_0)\) and \((\hat{\gamma}^{(k)}_n,\hat{\eta}^{(k)}_n)\). 

Suppose that we could show that
\begin{align}
\label{eq:condition1}
&n^{-1}\frac{\partial}{\partial(\gamma,\eta^T)^T}U_c^{(n)}(k,\gamma,\eta,\sigma_0^2)\bigg\rvert_{\gamma=\gamma^*_n, \eta = \eta^*_n} \xrightarrow{p}\mathbf{A}^{(k)}\hspace{2cm}\text{and} \\
\label{eq:condition2}
&-n^{-\frac{1}{2}}\begin{bmatrix}U_c(1, \gamma,\eta,\sigma^2) \\ \vdots \\ U_c(K, \gamma,\eta,\sigma^2) \end{bmatrix} \xrightarrow{d} N\left( \begin{bmatrix} 0 \\ \vdots \\ 0\end{bmatrix} , \begin{bmatrix} \mathbf{B}^{(1)} & \cdots & \mathbf{B}^{(1)} \\ \vdots & \ddots & \vdots \\ \mathbf{B}^{(1)} & \cdots & \mathbf{B}^{(K)}\end{bmatrix} \right),
\end{align}
then by an application of Slutsky's Theorem, we have the desired result. 

A simple application of the triangle inequality proves that condition~\eqref{eq:condition1} holds. The proof of convergence in probability closely follows the standard results for survival data seen in Theorem~VII.2.2 by~\citet{andersen1982cox} and the details of this step are found in the supplementary materials.

To prove condition~\eqref{eq:condition2} holds, first note that Expression~\eqref{eq:cond_group_a} is a sum over \(n\) independent and identically distributed terms. Then the multivariate Central Limit Theorem can be applied to the vector \((U_c(1,\gamma,\eta,\sigma^2), \dots, U_c(K,\gamma,\eta,\sigma^2))\) to establish asymptotic normality. It then remains to show that \(Cov(U_c(k_1,\gamma,\eta,\sigma^2), U_c(k_2,\gamma,\eta,\sigma^2))\xrightarrow{p} \mathbf{B}^{(k_1)}\) for \(k_1\leq k_2.\)

We now follow a similar structure to the partial likelihood function for survival data given by \cite{jennison1997group} and we create a new counting process that allows the conditional score statistic to be written as the sum of distinct increments. This counting process is defined by \(DN_i(k,t)=N_i(k,t)-N_i(k-1,t)\) for \(k=1,\dots,K\).The corresponding compensated counting process is therefore given by
\[DM_i(k,t) = DN_i(k,t)-\int_0^th_0(u)E_{0i}(t,\gamma,\eta,\sigma^2)(Y_i(k,u)-Y_i(k-1,u))du\]
for \(k=1,\dots,K\) and \(DM_i(0,t) = 0.\) The event for patient \(i\) can only occur in one interval, and therefore we have \(N_i(k,t)=\sum_{l=1}^k DN_i(l,t)\). For consistency with~\citet{andersen1982cox} and ~\citet{jennison1997group}, let Expression~\eqref{eq:cond_group_a} be denoted by \(W^{(n)}_j(k, \tau_k, \gamma,\eta,\sigma^2)\). Then by the definition of the difference in counting process, we have
\[W^{(n)}_j(k, \tau_k, \gamma,\eta,\sigma^2) = \int_0^\tau \sum_{i=1}^n  \sum_{l=1}^kH_{ij}^{(n)}(l,u, \gamma,\eta,\sigma^2)dDM_i(l,u)\]
where
\[H_{ij}^{(n)}(l,u,\gamma,\eta,\sigma^2) = n^{-\frac{1}{2}}(\{S_i(k,u, \gamma,\sigma^2),Z_i\}^T_j-e_c(l,u,\gamma,\eta,\sigma^2)_j)\]
and \(W^{(n)}_j(k, \tau_k, \gamma,\eta,\sigma^2)\) has expectation zero by a simple manipulation of Equation~\eqref{eq:exp_0}.

 The main result, with details given  in the supplementary materials, is as follows
\begin{align}
\notag
&Cov\left(W^{(n)}_{j_1}(k_1,\tau_{k_1}, \gamma_0,\eta_0,\sigma_0^2),W^{(n)}_{j_2}(k_2,\tau_{k_2}, \gamma_0,\eta_0,\sigma_0^2)\right) \\
\label{eq:cov}
=&\mathbb{E}\left(\sum_{i=1}^n\int_0^{\tau_{k_1}}H_{ij_1}^{(n)}(k_1,u, \gamma_0,\eta_0,\sigma_0^2)H_{ij_2}^{(n)}(k_1,u, \gamma_0,\eta_0,\sigma_0^2)h_0(u)E_{0i}(u,\gamma_0,\eta_0,\sigma_0^2)Y_i(k_1,u)du\right).
\end{align}

To complete the proof, it remains to show that Equation~\eqref{eq:cov} converges in probability to \(\mathbf{B}^{(k_1)}\) and we shall do so by using the definitions of the objects \(s_c^{(j)}\) and \(c^{(j)}\) for \(j=1,0,1,2\) and \(e_c,v^{(1)}_c\) and \(v^{(2)}\) given in Appendix~1. We are assuming that these limits exist by Conditions~\ref{conditions:conditions_joint_group} and we shall exploit the relationships between these terms. This working is exactly the same as for standard survival data given by~\citet{andersen1982cox}, with the details given in the supplementary materials, and we see that
\begin{align*}
&Cov\left(W^{(n)}_{j_1}(k_1,\tau_{k_1}, \gamma_0,\eta_0,\sigma_0^2),W^{(n)}_{j_2}(k_2,\tau_{k_2}, \gamma_0,\eta_0,\sigma_0^2)\right) \\
\xrightarrow{p} &\left(\int_0^\infty v^{(1)}_c(k_1,u, \gamma_0,\eta_0,\sigma^2_0)s^{(0)}_c(k_1,u, \gamma_0,\eta_0,\sigma^2_0)h_0(u)du\right)_{j_1j_2} = B^{(k_1)}_{j_1j_2}.
\end{align*}

By the argument that the behaviour of the estimates \(\hat{\gamma}^{(k)}\) and \(\hat{\eta}^{(k)}\) is dictated by Expression~\eqref{eq:cond_group_a}, we have
\[Cov\left(U_c(k_1,\gamma_0,\eta_0,\sigma^2_0),U_c(k_2,\gamma_0,\eta_0,\sigma^2_0)\right) \xrightarrow{p} \mathbf{B}^{(k_1)} \hspace{1cm}\text{ for } k_1 \leq k_2\]
which is the result.
\end{proof}

Similar to the fixed sample case, we have assumed that \(\sigma^2_0\) is known in the derivation of the distribution of \(\hat{\gamma}^{(k)}_n\) and \(\hat{\eta}^{(k)}_n.\) This is not generally the case but by arguments in~\citet{carroll2006measurement} Section A.3.3, we can find a consistent estimate to replace \(\sigma^2_0\) with in the group sequential conditional score function. At analysis \(k\) this estimate is given by
\begin{equation}
\label{eq:sigma_group}
\hat{\sigma}^{(k)2}=\frac{\sum_{i=1}^{n} \mathbb{I}\{m_i(k)>2\}R_i(k)}{\sum_{i=1}^{n}\mathbb{I}\{m_i(k)>2\}(m_i(k)-2)},
\end{equation}
where \(R_i(k)\) is the residual sum of squares for the least squares fit to all \(m_i(k)\) observations for patient \(i\) available at analysis \(k\).

\subsection{A group sequential trial design based on the conditional score analysis}
We shall use Theorem~\ref{theorem:joint_group} to create a group sequential test based on the joint model. Let \(\gamma_0,\eta_0\) and \(\sigma^2_0\) be the true values of the parameters \(\gamma,\eta\) and \(\sigma^2\) respectively. Using the group sequential conditional score method, let \(\hat{\gamma}^{(k)}, \hat{\eta}^{(k)}\) be the values of the parameters \(\gamma\) and \(\eta\) such that \(U_c(k,\gamma,\eta,\sigma^2)=0\) where the conditional score function is calculated using Equation~\eqref{eq:cond_score_k}. Further, let \(\hat{\sigma}^{(k)2}\) be the estimate for \(\sigma^2\) given in Equation~\eqref{eq:sigma_group}. By Theorem~\ref{theorem:joint_group}, for each \(k=1,\dots,K\), the marginal distribution of the parameter \(\hat{\eta}^{(k)}\) is
\[\sqrt{n}(\hat{\eta}^{(k)}-\eta_0) \xrightarrow{d} N(0, \Sigma^{(k)}_{22})\]
where 
\[\Sigma^{(k)}=(A^{(k)})^{-1}B^{(k)}((A^{(k)})^{-1})^T\]
and the matrices \(A^{(k)}\) and \(B^{(k)}\) are defined by Equations~\eqref{eq:group_A}--\eqref{eq:group_B}. Note that the subscript notation in the covariance matrix represents that the parameter \(\eta\) is the second parameter in the vector \((\gamma,\eta)\). In the more general case where \(Z_i\) is a vector of dimension \(p\) including other covariates appart from treatment indicator, then the information would be indexed by the corresponding dimension of the vector \((\gamma, \eta)^T\) which relates to the treatment effect. The matrices \(A^{(k)}\) and \(B^{(k)}\) are estimated using
\begin{align}
\label{eq:A_hat_k}
\hat{A}^{(k)} &= \frac{1}{n}\sum_{i=1}^{n} \int_0^\infty\left[\Gamma_i(k,u,\hat{\sigma}^{(k)2})-V_c^{(1)}(k,u,\hat{\gamma}^{(k)},\hat{\eta}^{(k)},\hat{\sigma}^{(k)2})+V_c^{(2)}(k,u,\hat{\gamma}^{(k)},\hat{\eta}^{(k)},\hat{\sigma}^{(k)2})\right]dN_i(k,u) \\
\label{eq:B_hat_k}
\hat{B}^{(k)} &= \frac{1}{n}\sum_{i=1}^{n} \int_0^\infty\left[V_c^{(1)}(k,u,\hat{\gamma}^{(k)},\hat{\eta}^{(k)},\hat{\sigma}^{(k)2})\right]dN_i(k,u).
\end{align}

The information matrix at analysis \(k\) of the group sequential trial is given by
\[\mathcal{I}_k=\frac{1}{n}\left[ (\hat{A}^{(k)})^{-1}\hat{B}^{(k)}((\hat{A}^{(k)})^{-1})^T\right]^{-1}_{22}.\]
Further, a standardised test statistic at analysis \(k\) is given by
\[Z_k=\hat{\eta}^{(k)}\sqrt{\mathcal{I}_k}.\]

\section{Designing group sequential trials when the canonical joint distribution does not hold}
\label{sec:non_canonical}

\subsection{Simulation evidence that the trial is conservative with respect to type 1 error rate}
\label{subsec:non_canonical_simulation}

We discuss the implications for a group sequential test when the sequence of test statistics does not follow the usual canonical joint distribution and we shall show that in some scenarios, it is appropriate to make this assumption anyway and proceed with the trial as though the canonical joint distribution holds anyway. Let \(\hat{\theta}^{(1)},\dots,\hat{\theta}^{(K)}\) be the sequence of treatment effect estimates in a group sequential trial from data available at analyses \(1,\dots,K\) respectively and \(\mathcal{I}_1, \dots, \mathcal{I}_K\) are the associated observed information levels. The ``canonical joint distribution" for the sequence of estimates \(\hat{\theta}^{(1)},\dots, \hat{\theta}^{(K)}\) is such that
\begin{enumerate}
\item \((\hat{\theta}^{(1)},\dots,\hat{\theta}^{(K)})\) is multivariate normal
\item \(\hat{\theta}^{(k)}\sim  N(\theta,\mathcal{I}_k^{-1}), \hspace{10mm} 1\leq k\leq K\)
\item \(Cov(\hat{\theta}^{(k_1)},\hat{\theta}^{(k_2)})=\mathcal{I}_{k_2}^{-1},  \hspace{10mm} 1\leq k_1 \leq k_2 \leq K.\)
\end{enumerate}
Supposing that the canonical joint distribution holds for the sequence of treatment effect estimates in the joint model, then the library of standard group sequential designs, such as \citet{pocock1977group} and \citet{o1979multiple}, could be directly applied in our setting 

In Section~\ref{sec:joint} we saw that asymptotically the sequence of treatment effect estimates in a group sequential trial obtained from the joint model is multivariate normally distributed. Further, each of these estimates is asymptotically unbiased. The first two conditions for the canonical distribution of the sequence of test statistics are satisfied. However, for the joint model, we have shown that 
\begin{equation}
\label{eq:joint_treat_var}
Var(\hat{\eta}^{(k)}) = \left[(A^{(k)})^{-1}B^{(k)}(A^{(k)})^{-1}\right]_{22} \text{ for } k=1,\dots,K 
\end{equation}
and
\begin{equation}
\label{eq:joint_treat_cov}
Cov(\hat{\eta}^{(k_1)}, \hat{\eta}^{(k_2)}) = \left[A^{(k_1)})^{-1}B^{(k_1)}(A^{(k_2)})^{-1}\right]_{22} \text{ for } k_1 < k_2.
\end{equation}
This implies that the third condition for the canonical joint distribution, which is the Markov property, is not satisfied. In this section, we discuss the implications of performing a group sequential trial when the assumption of a canonical joint distribution fails. We first show that there are some small differences between the matrices \((A^{(k)})^{-1}B^{(k)}\) and the identity matrix and describe why this difference is important. Then, we discuss some alternative methods which aim to correct for this violation of the Markov property.  In method 1, the trial is performed acting as though the canonical joint distribution holds, and we present some theory that this method controls the type 1 error rate conservatively.

The theory presented is for the case when a non-binding futility boundary is used which is where stopping for futility at an interim analysis is not mandatory as described in \citet{guidance2018adaptive}. FDA guidance recommends using non-binding futility rules because if binding rules are employed and not followed, then type 1 error rates are inflated. The calculation of the type 1 error rate therefore only depends on the upper boundary \(b_1, b_2\) which is depicted in Figure~\ref{fig:boundaries}.

The problem of not obtaining the canonical joint distribution is not unique to the conditional score method and the proof that this assumption holds is not always trivial. \citet{slud1982two} design a group sequential test which uses the modified-Wilcoxon statistic for two-sample survival data. The authors show that the increments in test statistics are correlated and their alternative proposed method has a similar structure to our method 2. Then, \citet{gu1993sequential} present a score process for the analysis of regression data under general right censorship and they implement the repeated significance method of \citet{slud1982two} for the sequential analysis of such data. Further, \citet{cook1996interim} discuss a modification of the error spending approach (also known as \citet{gordon1983discrete}) for sequences of treatment effect estimates that do not have the independent increments structure.

If the relationship \(A^{(k)}=B^{(k)}\) holds for each \(k=1,\dots,K,\) then
\[Cov(\hat{\eta}^{(k_1)},\hat{\eta}^{(k_2)})=\left[(A^{(k_2)})^{-1}\right]_{p+1,p+1}=Var(\hat{\eta}^{(k_2)})\]
and the third condition holds. Therefore, we shall measure the magnitude of the violation by considering the matrix \((A^{(k)})^{-1}B^{(k)}\) and to what extent it differs from the identity matrix. By definition, the variance of the estimate \(\hat{\eta}^{(k)}\) is found in the bottom right entry of the matrix \((\hat{A}^{(k)})^{-1}\hat{B}^{(k)}(\hat{A}^{(k)})^{-1}\) and hence, the bottom row of \((\hat{A}^{(k)})^{-1}\hat{B}^{(k)}\) is of interest here. 

We can find estimates \(\hat{A}^{(k)}\) and \(\hat{B}^{(k)}\) from simulated data. We have done this using a large sample size of n=4800 patients to reduce noise in these estimates. This calculation is computationally expensive and takes roughly two hours to compute for this value of \(n\). This is appropriate because, although both matrices depend on the sample size \(n\), they can each be written in the form \(\hat{A}^{(k)}=(1/n) \sum_{i=1}^n X_i(k,\gamma,\eta)\) and \(\hat{B}^{(k)}=(1/n)\sum_{i=1}^nY_i(k,\gamma,\eta)\) for some functions \(X_i(k,\gamma,\eta)\) and \(Y_i(k,\gamma,\eta).\) Therefore, in the formula \((\hat{A}^{(k)})^{-1}\hat{B}^{(k)}\), the value of \(n\) cancels out, and we are left with a function that converges in distribution to \((A^{(k)})^{-1}B^{(k)}\) as \(n\rightarrow \infty.\) Further, to reduce simulation error, the true values of \(\gamma\) and \(\eta\) are used in this calculation, which is appropriate because of consistency of the estimates \(\hat{\gamma}\) and \(\hat{\eta}\). 

Table~\ref{tbl:A_invB_null} shows the matrix \((\hat{A}^{(1)})^{-1}\hat{B}^{(1)}\) for \(\eta=0\) and different values of \(\gamma\) and \(\sigma^2.\) We have chosen to investigate the properties of this matrix at the first analysis because we see empirically that the majority of problems occur at early interim analyses. We simulated a data set with parameter values \(\gamma=0,0.03,0.06,0.09,\sigma^2=0,1,10,100\) and \(\eta=0.\) 
\begin{table}
\centering
\begin{threeparttable}
  \begin{tabular}{lcccc}
&\multicolumn{4}{c}{$\bs{\gamma}$} \\
$\mathbf{\sigma^2}$ & $\mathbf{0}$  & $\mathbf{0.03}$  & $\mathbf{0.06}$  & $\mathbf{0.09}$ \\[5pt]$\mathbf{0}$&$\left(\begin{array}{cc}1.00&0.00\\0.00&1.00\end{array}\right)$&$\left(\begin{array}{cc}1.00&0.00\\0.00&1.00\end{array}\right)$&$\left(\begin{array}{cc}1.00&0.00\\0.00&1.00\end{array}\right)$&$\left(\begin{array}{cc}1.00&0.00\\0.00&1.00\end{array}\right)$\\\rule{0pt}{4.5ex}
$\mathbf{1}$&$\left(\begin{array}{cc}1.00&0.00\\0.00&1.00\end{array}\right)$&$\left(\begin{array}{cc}1.01&0.00\\0.00&1.00\end{array}\right)$&$\left(\begin{array}{cc}1.01&0.00\\0.00&1.00\end{array}\right)$&$\left(\begin{array}{cc}1.02&0.00\\0.00&1.00\end{array}\right)$\\\rule{0pt}{4.5ex}
$\mathbf{10}$&$\left(\begin{array}{cc}1.03&0.00\\0.01&1.00\end{array}\right)$&$\left(\begin{array}{cc}1.06&0.00\\0.00&1.00\end{array}\right)$&$\left(\begin{array}{cc}1.12&0.00\\-0.02&1.00\end{array}\right)$&$\left(\begin{array}{cc}1.22&0.00\\0.00&1.00\end{array}\right)$\\\rule{0pt}{4.5ex}
$\mathbf{100}$&$\left(\begin{array}{cc}1.28&0.00\\-0.10&1.00\end{array}\right)$&$\left(\begin{array}{cc}1.63&0.00\\-0.47&1.00\end{array}\right)$&$\left(\begin{array}{cc}2.32&0.00\\-0.01&1.00\end{array}\right)$&$\left(\begin{array}{cc}3.49&0.00\\1.00&1.00\end{array}\right)$
\end{tabular}
  \begin{tablenotes}
  \item Parameter values \(\gamma=0,0.03,0.06,0.09\) and \(\sigma^2=0,1,10,100\) in the joint model  simulated with 4800 patients.
  \end{tablenotes}
\end{threeparttable}
\caption{Matrix \((\hat{A}^{(1)})^{-1}\hat{B}^{(1)}\) for the null hypothesis \(\eta = 0\).}
\label{tbl:A_invB_null}
  \end{table}

The matrices \(A^{(k)},\hat{A}^{(k)},B^{(k)}\) and \(\hat{B}^{(k)}\) are each of dimension \(2\times 2.\) The function \(V^{(2)}(k,u,\gamma,\eta)\) is such that \(\left[V^{(2)}(k,u,\gamma,\eta)\right]_{12}= \left[V^{(2)}(k,u,\gamma,\eta)\right]_{22} = 0\), and hence by Equations~\eqref{eq:A_hat_k} and~\eqref{eq:B_hat_k}  it can be shown that \([\hat{A}^{(k)}]_{12}=[\hat{B}^{(k)}]_{12}\) and \([\hat{A}^{(k)}]_{22}=[\hat{B}^{(k)}]_{22}.\) Further simple algebraic manipulation gives \([(\hat{A}^{(k)})^{-1}\hat{B}^{(k)}]_{12}=0\) and \([(\hat{A}^{(k)})^{-1}\hat{B}^{(k)}]_{22}=1\) exactly, which is shown in Table~\ref{tbl:A_invB_null}. The fact that \([(\hat{A}^{(1)})^{-1}\hat{B}^{(1)}]_{11}\) is a long way from \(1\) for \(\sigma^2=10\) and \(\sigma^2=100\) is therefore not a problem. As \(\sigma^2\) increases, the absolute value of \([(\hat{A}^{(1)})^{-1}\hat{B}^{(1)}]_{21}\) increases, but the value of \(\gamma\) has a small impact on the value of \([(\hat{A}^{(1)})^{-1}\hat{B}^{(1)}]_{21}.\) Thus, we may expect large values of \(\sigma^2\) to affect the achieved type 1 error rate.


\subsection{Method 1 - canonical joint distribution assumed}
\label{subsec:non_canonical_proof}
We consider alternative methods for creating a group sequential trial when the canonical joint distribution does not hold. In the first method, we construct the group sequential test by estimating \(Var(\hat{\eta}^{(k)}),k=1,\dots,K\) from the data and supposing \(Cov(\hat{\eta}^{(k_1)},\hat{\eta}^{(k_2)})\) for \(k_1<k_2\) are as specified in the canonical joint distribution. We shall prove that we have type 1 error rate less than \(\alpha\) and we also show, through simulation, that this method performs satisfactorily in practice with error rates diverging minimally from planned significance and power.

For this proof, we consider the case where \(K=2\) and the futility boundary is non-binding. We present a sketch proof in the supplementary materials for the case \(K=3\) and we believe that the results generalise for cases \(K > 3.\) To prove that the type 1 error rates are conservative, we shall compare the probabilities of crossing the boundaries of a group sequential trial for two sequences of treatment effect estimates; one where the canonical joint distribution does not hold and one where this assumption does hold. For a group sequential trial with \(K=2\), suppose that \(\hat{\eta}_1,\hat{\eta}_2\) are the sequence of treatment effect estimates that are calculated using the conditional score method. Let the true variance-covariance matrix for this sequence of estimates be \(\Sigma\). Proceeding using method 1, we let the information levels be calculated as \(\mathcal{I}_k = (\Sigma_{kk})^{-1}\) and the \(Z-\)statistic is given by \(Z_k=\hat{\eta}_k\sqrt{\mathcal{I}_k}\) for \(k=1,2.\) Under \(H_0\) we have
\begin{equation}
\label{eq:Z_dist}
\begin{bmatrix} \hat{\eta}_1 \\ \hat{\eta}_2\end{bmatrix} \sim N_K\left[ \left(\begin{array}{c} 0 \\ 0 \end{array}\right),\left( \begin{array}{cc}
\Sigma_{11} & \Sigma_{12}  \\
\Sigma_{12} & \Sigma_{22}
\end{array} \right) \right], \hspace{1cm}
\begin{bmatrix} Z_1 \\ Z_2\end{bmatrix} \sim N_K\left[ \left(\begin{array}{c} 0 \\ 0 \end{array}\right), \left( \begin{array}{cc}
1 & \rho\\
\rho & 1
\end{array} \right) \right]
\end{equation}
where the correlation parameter \(\rho\) is given by
\begin{equation}
\label{eq:rho}
\rho=Cov(Z_1,Z_2)=\frac{\Sigma_{12}}{\sqrt{\Sigma_{11}\Sigma_{22}}}.
\end{equation}

Suppose instead, that we have a different sequence of treatment effect estimates \(\hat{\eta}^*_1,\hat{\eta}^*_2\) with distribution given below. The values of \(\Sigma_{11}\) and \(\Sigma_{22}\) are the same as in~\eqref{eq:Z_dist} and using the same information levels given by \(\mathcal{I}_k = (\Sigma_{kk})^{-1}\), we define the standardised statistics \(Z^*_k=\hat{\eta}^*_k\sqrt{\mathcal{I}_k}\) for \(k=1,\dots,K.\) The joint distribution of the sequence \(\hat{\eta}^*_1,\hat{\eta}^*_2\) and the distribution of the \(Z-\)statistics are given by
\begin{equation}
\label{eq:Z_star_dist}
\begin{bmatrix} \hat{\eta}_1^* \\ \hat{\eta}_2^* \end{bmatrix} \sim N_K\left[ \left(\begin{array}{c} 0 \\ 0 \end{array}\right),\left( \begin{array}{cc}
\Sigma_{11} & \Sigma_{22}  \\
\Sigma_{22} & \Sigma_{22}
\end{array} \right) \right], \hspace{1cm}
\begin{bmatrix} Z_1^* \\ Z_2^* \end{bmatrix} \sim N_K\left[ \left(\begin{array}{c} 0 \\ 0 \end{array}\right), \left( \begin{array}{cc}
1 & \rho^*  \\
\rho^* & 1
\end{array} \right) \right]
\end{equation}
with correlation parameter
\begin{equation}
\label{eq:rho_star}
\rho_{12}^*=Cov(Z_1^*,Z_1^*)=\sqrt{\frac{\Sigma_{22}}{{\Sigma_{11}}}}.
\end{equation}
This sequence of treatment effect estimates \(\hat{\eta}^*_1,\hat{\eta}^*_2\) therefore has the canonical joint distribution and \(Z^*_1,Z^*_2\) has the canonical joint distribution for a sequence of \(Z\)-statistics, with information levels \(\mathcal{I}_k=1/\Sigma_{kk}\) for \(k=1,\dots,K\).

The upper boundary points \(b_1,b_2\) are calculated under the assumption that the canonical joint distribution holds so as to give a group sequential test with the correct type 1 error rate \(\alpha\). Hence, we have that 
\[\mathbb{P}_{\eta=0}(Z^*_1  > b_1 \cup Z^*_2 > b_2) =\alpha.\]
We consider the probability of rejecting \(H_0\) when we apply this boundary to the sequence \(Z_1,Z_2.\) We aim to prove that 
\begin{equation}
\label{eq:Z_alpha}
\mathbb{P}_{\eta=0}(Z_1  > b_1 \cup Z_2 > b_2) \leq \alpha.
\end{equation}

The following conditions are needed for the proof that type 1 error rate is conservative. 
\begin{conditions}
\label{conditions:efficiency}\
\begin{enumerate}
\item The upper boundary of a group sequential trial, on the \(Z-\)scale, is such that
\(b_1 \geq b_2 \geq 0.\)
\item We have 
\(\Sigma_{12} \geq \Sigma_{22}.\)
\end{enumerate}
\end{conditions}
Condition 1 of Conditions~\ref{conditions:efficiency} holds under common error spending functions with increasing information sequences, which are most often used in practice. Condition 2 of Conditions~\ref{conditions:efficiency} should be checked by simulation before proceeding with the analysis. To do so, the investigator would choose sensible values for all the parameters in the joint model, simulate a large dataset of \(4800\) patients using these parameter values, and calculate an estimate for the variance-covariance matrix \(\Sigma\) for the sequence of estimates \(\hat{\eta}^{(1)},\dots,\hat{\eta}^{(K)}.\) We have found that calculations for various examples has always lead to condition 2 being satisfied. Further, the scenarios that we have checked span a 3-dimensional grid of \(\eta,\gamma\) and \(\sigma^2\) values each ranging from small to large and hence, we believe that the scenarios we have checked span a suitable range of the parameter values. In the rare event that this condition does not hold, a solution is to employ method 2, which will be described later. It can also be seen by simple algebraic manipulation that condition 2 implies \(\rho\geq \rho^*.\)

The following theorem shows that Equation~\eqref{eq:Z_alpha} holds. 
\begin{theorem}
\label{theorem:efficiency}
Let \(Z_1,Z_2\) be the standardised statistics of a group sequential trial with distribution given by~\eqref{eq:Z_dist} and let \(Z^*_1,Z^*_2\) be the statistics with distribution given by~\eqref{eq:Z_star_dist}. Let \(\alpha\) be the planned type 1 error rate and suppose that \(b_1,b_2\) are the upper boundary points on the \(Z\)-scale such that
\[\mathbb{P}_{\eta=0}(Z^*_1  > b_1 \cup Z^*_2 > b_2 ) =\alpha.\]
Suppose that Conditions~\ref{conditions:efficiency} hold. Then the Type 1 error rate when applying the boundary for \(Z^*_1,\dots,Z^*_K\) to \(Z_1,\dots,Z_K\) is
\[\mathbb{P}_{\eta=0}(Z_1  > b_1 \cup Z_2 > b_2) \leq \alpha.\]
\end{theorem}
\begin{proof}
The problem is equivalent to proving that 
\[\mathbb{P}_{\eta=0}(Z_1  > b_1 \cup Z_2 > b_2) \leq \mathbb{P}_{\eta=0}(Z^*_1  > b_1 \cup Z^*_2 >).\]
and by another representation for the above probabilities, we aim to show that
\begin{align*}
&\mathbb{P}_{\eta=0}(Z_1>b_1) + \mathbb{P}_{\eta=0}(Z_1>b_2) - \mathbb{P}_{\eta=0}(Z_1>b_1\cap Z_2 > b_2)\leq \mathbb{P}_{\eta=0}(Z_1>b_1 \cap Z_2 > b_2)  \\
\leq & \mathbb{P}_{\eta=0}(Z^*_1>b_1) + \mathbb{P}_{\eta=0}(Z^*_1>b_2) - \mathbb{P}_{\eta=0}(Z^*_1>b_1\cap Z^*_2 > b_2)\leq \mathbb{P}_{\eta=0}(Z_1>b_1 \cap Z_2 > b_2).
\end{align*}
Under \(\eta=0\), the marginal distributions of \(Z_1,Z_2,Z^*_1\) and \(Z^*_2\) are equivalent and are all \(N(0,1)\) random variables and hence the probabilities are such that \(\mathbb{P}_{\eta=0}(Z_k > b_k)=\mathbb{P}_{\eta=0}(Z^*_k>b_k)\) for each \(k=1,2.\) Therefore, the problem is reduced to showing that 
\begin{equation}
\label{eq:K2}
\mathbb{P}_{\eta=0}(Z^*_1>b_1\cap Z^*_2 > b_2)\leq \mathbb{P}_{\eta=0}(Z_1>b_1 \cap Z_2 > b_2).
\end{equation}
when \(\rho^*\leq \rho.\)

In the below calculations, we appeal to the fact that for two random variables which are bivariate normally distributed, the conditional distribution of one normal random variable on the other normal random variable is also normal. Specifically, we have that \(Z_2|Z_1=z_1\sim N(\rho z_1, 1-\rho^2)\). Let \(\phi(\cdot)\) and \(\Phi(\cdot)\) denote the probability density function and cumulative distribution function of a standard normal random variable respectively, then the probability on the left hand side in Equation~\eqref{eq:K2} is 
\[\mathbb{P}_{\eta=0}(Z_1 >b_1 \cap Z_2>b_2) = \int_{b_1}^\infty \left[1-\Phi\left(\frac{b_2-\rho z_1}{\sqrt{1-\rho^2}}\right)\right]\phi(z_1)dz_1.\]

The corresponding calculation where \(Z^*_1\) and \(Z^*_2\) replace \(Z_1\) and \(Z_2\) yields the following
\[\mathbb{P}_{\eta=0}(Z^*_1>b_1 \cap Z^*_2 > b_2) =\int_{b_1}^\infty \left[1-\Phi\left(\frac{b_2-\rho^* z_1}{\sqrt{1-{\rho^*}^2}}\right)\right]\phi(z_1)dz_1\]
and since \(\Phi(\cdot)\) is strictly increasing, it suffices to show that whenever \(z_1>b_1\), then
\begin{equation}
\label{eq:K2_2}
\frac{b_2-\rho z_1}{\sqrt{1-{\rho}^2}}\leq\frac{b_2-\rho^* z_1}{\sqrt{1-{\rho^*}^2}}.
\end{equation}

We have by assumption that \(\rho^*\leq \rho.\) Further note that by definition \(0\leq \rho \leq 1\) and \(0\leq \rho^*\leq 1.\) The following shows a simple algebraic manipulation of this inequality, which gives
\[\rho^* \leq \rho \iff \frac{\sqrt{1-{\rho^*}^2}-\sqrt{1-\rho^2}}{\rho\sqrt{1-{\rho^*}^2}-\rho^*\sqrt{1-\rho^2}}\leq 1.\]

Finally, using the above inequality and Conditions~\ref{conditions:efficiency} that \(0\leq b_2 \leq b_1\), we have for \(z_1\geq b_1\) that
\[b_2 \frac{\sqrt{1-{\rho^*}^2}-\sqrt{1-\rho^2}}{\rho\sqrt{1-{\rho^*}^2}-\rho^*\sqrt{1-\rho^2}} \leq b_2\leq b_1\leq z_1\]
and a simple rearrangement shows that Equation~\eqref{eq:K2_2} is satisfied.
\end{proof}

Table~\ref{tbl:method1} show estimates of type 1 error rate. This was by simulating \(10^4\) data sets, each with a sample size \(n=365.\) This sample size was chosen because it gives power 0.9 for \(\gamma=0.06,\sigma^2=1,\phi=2.5\) and for other parameter values, the power was close to 0.9. For each data set, we calculate the estimates \(\hat{\eta}^{(1)},\dots,\hat{\eta}^{(K)}\) and estimates of the covariance matrices \(\Sigma^{(1)},\dots,\Sigma^{(K)}\) using estimates \(\hat{A}^{(1)},\dots,\hat{A}^{(K)},\hat{B}^{(1)},\dots,\hat{B}^{(K)}\) given by Equations~\eqref{eq:A_hat_k} and~\eqref{eq:B_hat_k}. The boundary points \(a_1,\dots,a_K\) and \(b_1,\dots,b_K\) are then calculated under the assumption that the canonical joint distribution holds and the type 1 error rate is calculated as the proportion of replicates that reject the null hypothesis. This simulation study is computationally expensive and with \(10^4\) replicates, there is noise in the simulation results. Taking this into account, the simulation results support the relevance of asymptotic theory since all empirical type 1 error rates are within 2 standard deviations of 0.025.
\begin{table}
\centering
\begin{threeparttable}
   \begin{tabular}{llrrrrrrrrr}
  \textbf{Method}&\(\bs{\phi}\)& \(\bs{\sigma^2}\) &\multicolumn{4}{c}{Type 1 error} &\multicolumn{4}{c}{Percentage of time method fails}\\
  &&&$\bs{\gamma}\mathbf{=0}$ &$\bs{\gamma}\mathbf{=0.03}$&$\bs{\gamma}\mathbf{=0.06}$&$\bs{\gamma}\mathbf{=0.09}$
  &$\bs{\gamma}\mathbf{=0}$ &$\bs{\gamma}\mathbf{=0.03}$&$\bs{\gamma}\mathbf{=0.06}$&$\bs{\gamma}\mathbf{=0.09}$
  \\[5pt]1&$\mathbf{2.5}$&$\mathbf{0}$&0.022&0.022&0.022&0.023&0.000&0.000&0.000&0.000\\1&$\mathbf{2.5}$&$\mathbf{1}$&0.024&0.025&0.026&0.023&0.000&0.000&0.000&0.000\\1&$\mathbf{2.5}$&$\mathbf{10}$&0.023&0.024&0.023&0.025&0.000&0.000&0.000&0.000\\1&$\mathbf{2.5}$&$\mathbf{100}$&0.023&0.028&0.026&0.024&0.000&0.000&0.000&0.000\\\rule{0pt}{3.5ex}2&$\mathbf{2.5}$&$\mathbf{0}$&0.022&0.022&0.022&0.023&0.000&0.000&0.000&0.000\\2&$\mathbf{2.5}$&$\mathbf{1}$&0.024&0.025&0.026&0.023&0.000&0.000&0.000&0.000\\2&$\mathbf{2.5}$&$\mathbf{10}$&0.023&0.024&0.023&0.025&0.040&0.100&0.120&0.060\\2&$\mathbf{2.5}$&$\mathbf{100}$&0.022&0.028&0.026&0.029&49.140&37.360&35.950&40.420\\\rule{0pt}{3.5ex}3&$\mathbf{2.5}$&$\mathbf{0}$&0.028&0.024&0.023&0.023&0.000&0.000&0.000&0.000\\3&$\mathbf{2.5}$&$\mathbf{1}$&0.028&0.026&0.027&0.023&0.000&0.000&0.000&0.000\\3&$\mathbf{2.5}$&$\mathbf{10}$&0.029&0.025&0.023&0.026&0.140&0.050&0.060&0.310\\3&$\mathbf{2.5}$&$\mathbf{100}$&0.024&0.028&0.027&0.029&1.930&1.950&1.840&1.850\\
  \end{tabular}
  \begin{tablenotes}
  \item Calculated using a simulation study with 365 patients and \(10^4\) replicates.
  \item Standard error is 0.0016 when the method does not fail.
  \end{tablenotes}
\end{threeparttable}
\caption{Type 1 error rates and percentage of times that each method fails.}
  \label{tbl:method1}
  \end{table}

\subsection{Method 2 - Use the complete structure of the covariance matrix}
The second method for dealing with estimates from the joint model does not rely on the canonical joint distribution assumption. Instead, the group sequential boundaries are calculated using the complete structure of the variance-covariance matrix for the sequence of treatment effect estimates across analyses. This differs from when the canonical joint distribution is assumed because in such a case, only the variances are required and the covariances are ignored. To calculate the boundary points of the group sequential test, we are required to calculate a \(K-\)dimensional integral over the joint probability distribution of the sequence of test statistics. This integration calculation can be performed numerically using the R package mvtnorm by \citet{genz2020package}.

However, this method poses some practical difficulties. For example, during the conditional score method, the variance-covariance matrix \(\Sigma\) in Equation~\eqref{eq:sigma} is estimated with error which can sometimes result in a non positive-definite estimate \(\hat{\Sigma}.\) In such a case, the boundary calculations cannot be performed. \citet{slud1982two} allude to a similar problem in their sequential analysis of modified-Wilcoxon scores for two-sample survival data. Table~\ref{tbl:method1} shows the percentage of times that this problem occurs. We see that for extremely noisey longitudinal data with \(\sigma^2=100\) this problem occurs roughly 40\% of the time and for \(\sigma^2=10\) we have the problem occurring infrequently. This is not a problem for small \(\sigma^2\) and in such a case, the type 1 error rates are as expected. In summary, this method makes no assumption about the covariance structure of the sequence of treatment effect estimates and therefore the type 1 error rate is preserved. 

\subsection{Method 3 - Create an asymptotically efficient estimate}
For the final method, we follow the approach of \citet{van2022combining} where a new estimator is created which reaches asymptotic efficiency as it is designed in such a way to minimise the variance. The efficient estimate at analysis \(k\) is a linear combination of the original estimates at analyses up to and including \(k\). We choose the weights of the linear combination using a Lagrange multiplier method in such a way that the variance is minimised. Ou rmotivation for this method follows the simple result by \citet{jennison1997group}, that all asymptotically efficient estimators have the canonical joint distribution. \citet{van2022combining} show that this new estimator has the correct canonical distribution, and hence the group sequential methods can be used without hesitation.

Similarly to method 2, there are some limitations to this method as it relies too heavily on accurately estimating the covariance matrix of the sequence of treatment effect estimates. \citet{van2022combining} do not observe these issues since their results focus on simulations with sample sizes much larger than our chosen \(n=365\). In some cases, because the variance-covariance matrix is estimated with error, it is possible to choose the weights of the Lagrange multiplier in such a way that the new estimate has negative variance. Table~\ref{tbl:method1} shows a similar pattern to method 2, that this numerical problem occurs when the variance of the longitudinal data is extreme so that \(\sigma^2=100\), occurs infrequently for \(\sigma^2=10\) and no problems occur for small \(\sigma^2.\) 

Given that \(Cov(\hat{\eta}^{(k_1)}, \hat{\eta}^{(k_2)})\) is very close to \(Var(\hat{\eta}^{(k_2)})\) for \(k_1< k_2\), there is not a lot to be gained by the more complex methods 2 and 3.  The calculations are sensitive to errors in estimates of \(Cov(\hat{\eta}^{(k_1)}, \hat{\eta}^{(k_2)})\) and \(Var(\hat{\theta}^{(k_2)})\). In some cases, the errors in these calculations lead to these methods simply not working. This raises doubts about how well they work in less extreme cases. Despite these problems, methods 2 and 3 perform adequately in simulation studies. However, this does not change our view that method 1 is preferable.

\section{Results}
\label{sec:results}

We aim to assess the efficiency gain when longitudinal data are included in the analysis compared to when this longitudinal data is available, yet ignored. In this case, we believe that our joint model is correctly specified and therefore, we shall simulate clinical trial data from the true joint model and analyse it in two separate ways. The first way is to fit the data to the joint model using the conditional score method to find a treatment effect estimate and the second way is to ignore the longitudinal data, fitting a Cox model to the survival data and find the maximum partial likelihood estimate of the treatment effect. We are interested in comparing the sample sizes required in each method to achieve the same power. A comparison of these sample sizes reflects the efficiency of the test incorporating the longitudinal data.

We shall simulate using the joint model in Equation~\eqref{eq:surv}. To fit the joint model to the data, we do not need to assume a distribution for the random effects, however we must specify this for simulation purposes and we shall simulate using the following
\[\begin{bmatrix} b_{i0} \\ b_{i1} \end{bmatrix}\sim N\left( \begin{bmatrix} \mu_0 \\ \mu_1 \end{bmatrix}, \begin{bmatrix} \phi_0^2 & 0 \\ 0 & \phi_1^2 \end{bmatrix}\right).\]
Unless otherwise stated, throughout the simulation studies, we shall use parameter values
\begin{equation}
\label{eq:values1}
(\mu_0,\mu_1)=(6,3), \phi_0^2=12.25, \phi_1^2=6.25, \sigma^2=10, h_0(t)=5.5, \gamma=0.03, \lambda=0.022 \text{ and }\eta=-0.5.
\end{equation}
These parameter values are based on the aids dataset in the R package JM by~\citet{rizopoulos2010jm}.

We shall test the one-sided hypothesis
\[H^{(J)}_0:\eta_J\geq 0, \hspace{1cm} H^{(J)}_A:\eta_J< 0.\]
Here the subscript notation \(J\) represents that the parameter \(\theta_J\) is from the "joint" model. We fit the joint model using the conditional score method to find a treatment effect estimate \(\hat{\eta}_J\) in order to perform this hypothesis test. 

Our aim is to find the sample size, \(n_J\), required using the conditional score method to achieve Type 1 error rate \(\alpha=0.025\) when the true treatment effect is \(\eta_J=0\) and power \(1-\beta=0.9\) when \(\eta_J=-0.5\). An estimate of the sample size will be calculated by simulation. The trial is designed with 2 years recruitment and 3 years follow-up and the group sequential trial has analyses at \(20,30,40,50\) and \(60\) months.  When increasing the sample size, we do so by increasing the rate of recruitment so accrual and follow-up periods in the the trial design stay fixed. This is to ensure that differences in power are purely due to the sample size and not changes in the trial design as sample size increases.

The trial uses an error spending design given by \citet{gordon1983discrete}. That is, the boundary constants \(b_1,\dots,b_K\) are chosen to satisfy
\begin{align*}
\mathbb{P}_{\eta=0}(Z_1 > b_1) & = \min\{\alpha(\mathcal{I}_1/\mathcal{I}_{max})^2,\alpha\} \\
\mathbb{P}_{\eta=0}(Z_1 < b_1,\dots,Z_{k-1}<b_{k-1},Z_k>b_k)&=\min\{\alpha(\mathcal{I}_k/\mathcal{I}_{max})^2,\alpha\} -\min\{\alpha(\mathcal{I}_{k-1}/\mathcal{I}_{max})^2,\alpha\} \text{ for } k=2,\dots,K
\end{align*}
where the value of \(\mathcal{I}_{max}\) is calculated to ensure that the trial has power \(1-\beta\) when \(\eta=-0.5\) as described by~\citet{jennison2000group}.

Maximum sample sizes \(n_J\) are given in Table~\ref{tbl:n_r_values}. The first analysis, at 20 months, occurs just before the end of the recruitment period which is 2 years. Trials that terminate at the first interim analysis may recruit less than \(n_J\) patients, however this occurs with very small probability. Hence, the expected sample size will be very close to the maximum sample size for each model and therefore the maximum sample size is a useful measure to compare methods.  The sample size increases as \(\sigma^2\) increases. This reflects that noisy longitudinal data is associated with high variance or small information levels. Sample sizes are particularly high in each case where \(\sigma^2=100,\) which has been chosen as an extreme value. Further, sample sizes appear to increase slightly with \(\gamma\) and decrease slightly with \(\phi^2_1\). 

We now consider the analysis when the longitudinal data is ignored. We believe the joint model to be true and correct, however we shall fit the data to a Cox model. To do so, we shall simulate data from the joint model and then fit this data to a misspecified Cox proportional hazards model. The Cox model is given by:
\begin{equation}
\lambda_i(t)=\tilde{\lambda}_0(t)\exp\{\eta_CZ_i\}.
\end{equation}
For this clinical trial, we test the hypothesis
\begin{equation}
\label{eq:H_Cox}
H_0^{(C)}:\eta_C\geq 0, \hspace{1cm} H_A^{(C)}:\eta_C < 0
\end{equation}
and we find a treatment effect estimate \(\hat{\eta}_C\) using the maximum partial likelihood method as in \citet{jennison1997group}.

Although this model is misspecified, type 1 error is not affected. This is because, under \(H_0^{(J)}\) we have \(\eta_J=0\) and there is no difference between treatment groups in overall survival. When fitting the Cox model to the data, the longitudinal data trajectory is reflected in the function \(\tilde{\lambda}_0(t)\) so that  we also have that \(\eta_C=0.\) Hence, \(H_0^{(C)}\) is also true.

Let \(n_C\) be the sample size such that we achieve type 1 error \(\alpha=0.025\) when \(\eta_J=0\) and power \(1-\beta=0.9\) when \(\eta_J=-0.5\) when we perform the hypothesis test in~\eqref{eq:H_Cox}. Values of \(n_C\) and \(n_J\) are given in Table~\ref{tbl:n_r_values}. The values of \(\gamma\) and \(\phi_1^2\) for simulation are varied. Notice that \(n_C\) does not change with \(\sigma^2\) since this plays no role in simulating survival times, and the longitudinal data, which is affected by \(\sigma^2\), is ignored. As the value of \(\gamma\) increases, the sample size \(n_C\) increases. This represents that as the longitudinal data has more weight in the survival hazard rate, ignoring the longitudinal data results in an increasingly inefficient clinical trial. When \(\gamma=0,\) this represents the case where longitudinal data is available yet has no influence on the survival function. In this case, \(n_C < n_J\) and it is more efficient to fit the data to the simple Cox model. We see that the value of \(n_C\) increasess with \(\phi_1^2\). This is the variance amongst patients of the slopes of the longitudinal trajectories. This indicates that the simple Cox model is unable to account for large differences between individual patients.

To compare the sample sizes obtained using the joint model and the misspecified Cox model, we define ``relative efficiency" to be
\(RE=\frac{n_C}{n_J}.\)
Using this definition, when \(RE>1\) we interpret this as the joint model analysis being the more efficient model to use and similarly when \(RE<1,\) the Cox model analysis is the more efficient analysis method.

Table~\ref{tbl:n_r_values} shows the relative efficiency results. We see that RE increases with \(\gamma\), increases with \(\phi_1^2\) and remains constant with \(\sigma^2\) apart from the case where \(\sigma^2=100\) which reflects extremely noisy data. Also, we see that \(RE=0.97\) when \(\gamma=0\) and \(\sigma^2=100\) which indicates that when the longitudinal data is not correlated with the survival endpoint and the longitudinal data is noisy, the simple Cox model is a slightly more efficient method for estimating the treatment effect. Apart from the case where \(\gamma=0\), it is always more efficient to analyse the data using the joint modelling approach. Even when \(\gamma=0,\) fitting the data to the simple Cox model for survival data is only marginally more efficient than fitting the data to the joint model. In the extreme case, 2.63 times as many patients are required to analyse the data using the Cox model as when the joint modelling framework is used.
\begin{table}
\centering
\begin{threeparttable}
  \begin{tabular}{llcccccccccccc}
  \(\bs{\phi}\)& \(\bs{\sigma^2}\) &\multicolumn{3}{c}{ $\bs{\gamma}\mathbf{=0}$}   &\multicolumn{3}{c}{ $\bs{\gamma}\mathbf{=0.03}$} 
  &\multicolumn{3}{c}{ $\bs{\gamma}\mathbf{=0.06}$} &\multicolumn{3}{c}{ $\bs{\gamma}\mathbf{=0.09}$} \\
  && $n_C$  & $n_J$   & RE & $n_C$  & $n_J$   & RE& $n_C$  & $n_J$ 
  & RE& $n_C$  & $n_J$   & RE \\[5pt]$\mathbf{2.5}$&$\mathbf{0}$&363&363&1.00&421&365&1.16&528&364&1.45&607&365&1.67\\$\mathbf{2.5}$&$\mathbf{1}$&363&364&1.00&421&365&1.16&528&365&1.45&607&369&1.65\\$\mathbf{2.5}$&$\mathbf{10}$&363&364&1.00&421&364&1.16&528&365&1.45&607&375&1.62\\$\mathbf{2.5}$&$\mathbf{100}$&363&373&0.97&421&374&1.13&528&420&1.26&607&522&1.16\\\rule{0pt}{3.5ex}$\mathbf{0}$&$\mathbf{1}$&351&362&0.97&371&363&1.02&401&367&1.09&439&357&1.23\\$\mathbf{2.5}$&$\mathbf{1}$&363&364&1.00&421&365&1.16&528&365&1.45&607&369&1.65\\$\mathbf{5}$&$\mathbf{1}$&366&365&1.00&462&378&1.22&720&391&1.84&990&421&2.35\\$\mathbf{7.5}$&$\mathbf{1}$&380&372&1.02&501&404&1.24&804&416&1.93&1185&450&2.63
  \end{tabular}
  \begin{tablenotes}
  \item \(RE=n_C/n_J\), \(\gamma\) is the correlation parameter between the longitudinal and survival endpoints,\(\sigma^2\) is the measurement error of the longitudinal data and \(\phi_2\) is the variance of the random effects \(b_{1i},\dots, b_{1n}\).
  \end{tablenotes}
\end{threeparttable}
\caption{Maximum sample sizes required for power 0.9 when fitting the joint model and the misspecified Cox model to the data.}
\label{tbl:n_r_values}
 \end{table}

\section{Conclusions}
The conditional score method is used to find a treatment effect estimate in the joint model of longitudinal and time to event data and we have displayed new theoretical results for the distribution of the sequence of treatment effect estimates \(\hat{\eta}_1,\dots,\hat{\eta}_K\) found using the conditional score method in a group sequential trial. Although the canonical joint distribution for the sequence \(\hat{\eta}_1,\dots,\hat{\eta}_K\) does not hold, we show that it is sensible and practical to proceed assuming that the canonical joint distribution holds anyway. In particular, we have proven that by assuming the canonical joint distribution holds, and using a non-binding futility boundary, the trial is conservative  with respect to type 1 error rates. Finally, using simulation studies we have seen that the deviations from planned type 1 error \(\alpha\) are minimal. Other benefits of using the conditional score method are that no distributional assumptions are required for the random effects of the longitudinal data and the analysis is semi-parametric so that it is not necessary to estimate the baseline hazard function.

We have shown that by including the longitudinal data, compared to the case where the longitudinal data is observed but left out of the analysis, we can greatly improve the efficiency of the trial with respect to sample size. In some cases, 2.63 times as many patients are required to achieve the same power in the analysis where the longitudinal data is left out.
\appendix
\label{app:conditions}
\section*{Appendix 1}
\subsection*{Conditions for the asymptotic theory of parameter estimates}

To ensure the existence of the asymptotic covariance matrix \(\Sigma\), we require the probabilistic limits of \(S_c^{(0)}(k), \dots,S_c^{(2)}(k)\) and \(C^{(1)}(k),\dots,C^{(3)}(k)\) given by~\eqref{eq:SC} to exist. The limits are defined through the following conditions.
\begin{conditions}
\label{conditions:conditions_joint_group}\
\begin{enumerate}
\item There exist neighbourhoods \(\Gamma\) of \(\gamma_0\) and \(N\) of \(\eta_0\) and for each \(k=1,\dots,K\) there are functions \(s_c^{(0)}(k,t, \gamma, \eta,\sigma^2)\), \(s_c^{(1)}(k,t, \gamma, \eta,\sigma^2)\), \(s_c^{(2)}(k,t, \gamma, \eta,\sigma^2)\), \(c^{(1)}(k,t, \gamma, \eta,\sigma^2)\) and \(c^{(2)}(k,t, \gamma, \eta,\sigma^2)\) defined on \([0,\infty)\times\Gamma\times N\) such that
\begin{align*}
\sup_{t\in[0,\infty),\gamma\in\Gamma,\eta\in N}\left\| S_c^{(j)}(k,t, \gamma, \eta,\sigma^2)-s_c^{(j)}(k,t, \gamma, \eta,\sigma^2)\right\| &\xrightarrow{p} 0 \text{ for } j=0,1,2 \\
\sup_{t\in[0,\infty),\gamma\in\Gamma,\eta\in N}\left\| C^{(j)}(k,t, \gamma, \eta,\sigma^2)-c^{(j)}(k,t, \gamma, \eta,\sigma^2)\right\| &\xrightarrow{p} 0 \text{ for } j=1,2.
\end{align*}
\item Each \(s_c^{(j)}(k, t, \gamma, \eta,\sigma^2)\) and \(c^{(j)}(k, t, \gamma, \eta,\sigma^2)\) is a continuous function of \(\gamma \in \Gamma\) and \(\eta\in N\) uniformly in \(t\in [0,\infty),\) and bounded on \([0,\infty)\times\Gamma\times N\).
\item For each \(k=1,\dots,K\) \(s_c^{(0)}\) and \(c^{(0)}\) are bounded away from zero on \([0,\infty)\times\Gamma\times N\).
\end{enumerate}
\end{conditions}
It is clear that the probabilistic limits \(e_c(k,t,\gamma,\eta,\sigma^2)\) of \(E_c(k,t,\gamma,\eta,\sigma^2)\), \(v^{(1)}_c(k,t,\gamma,\eta,\sigma^2)\) of \(V^{(1)}_c(k,t,\gamma,\eta,\sigma^2)\) and \(v^{(2)}_c(k,t,\gamma,\eta,\sigma^2)\) of \(V^{(2)}_c(k,t,\gamma,\eta,\sigma^2)\) exist and can expressed in terms of \(s_c^{(j)}(k,t,\gamma,\eta,\sigma^2)\) and \(c^{(j)}(k,t,\gamma,\eta,\sigma^2)\) for \(j=0,1,2\) and these are
\[e_c = \frac{s_c^{(1)}}{ s_c^{(0)}} , \hspace{1cm}v^{(1)}_c= \frac{s_c^{(2)}}{ s_c^{(0)}}-\frac{s_c^{(1)}s_c^{(1)^T}}{[s_c^{(0)}]^2},\hspace{1cm}v^{(2)}_c = \frac{c^{(2)}}{ s_c^{(0)}}-\frac{s_c^{(1)}c^{(1)^T}}{[ s_c^{(0)}]^2}. \]

\section*{Software}
All statistical computing and analyses were performed using the software environment R version 4.0.2.
Programming code for sample size calculations, is
available at \texttt{https://github.com/abigailburdon/Conditional-score-GST}.

\section*{Acknowledgements}
This research was funded by the Engineering and Physical Sciences Research Council.

\bibliographystyle{plainnat}  
\bibliography{conditional_score}

\end{document}


\maketitle

\section{Differentiation of the conditional score}
We shall differentiate the conditional score function at analysis \(k\), \(U_c(k,\gamma,\eta,\sigma^2),\) which is given in Equation~(9) of the main paper, with respect to the vector \((\gamma,\eta)^T\). First we consider differentiating the function \(E_{0i}(t,\gamma,\eta,\sigma^2)\) with respect to \(\gamma\) and \(\eta\) separately. These are
\begin{align*}
\frac{\partial}{\partial\gamma} E_{0i}(t,\gamma,\eta,\sigma^2) &= \frac{\partial}{\partial\gamma} \exp\{\gamma S_i(t,\gamma,\sigma^2)-\gamma^2\sigma^2\theta_i(t)/2+\eta Z_i\} \\
&= \frac{\partial}{\partial\gamma} \exp\{\gamma^2\sigma^2\theta_i(t)dN_i(t)+\gamma\hat{X}_i(t)-\gamma^2\sigma^2\theta_i(t)/2+\eta Z_i\} \\
&=(2\gamma\sigma^2\theta_i(t)dN_i(t)+\hat{X}_i(t)-\gamma\sigma^2\theta_i(t)) E_{0i}(t,\gamma,\eta,\sigma^2) \\
&=(2S_i(t,\gamma,\sigma^2)-\hat{X}_i(t)-\gamma\sigma^2\theta_i(t)) E_{0i}(t,\gamma,\eta,\sigma^2) \\
\frac{\partial}{\partial\eta} E_{0i}(t,\gamma,\eta,\sigma^2) &= \frac{\partial}{\partial\eta} \exp\{\gamma S_i(t,\gamma,\sigma^2)-\gamma^2\sigma^2\theta_i(t)/2+\eta Z_i\} \\
& = Z_iE_{0i}(t,\gamma,\eta,\sigma^2).
\end{align*}
Derivatives of \(E_{0i}(t,\gamma,\eta,\sigma^2)\) are needed for differentiating the function \(S^{(0)}_c(t,\gamma,\eta,\sigma^2)\). Similarly for differentiation of the function \(S^{(1)}_c(t,\gamma,\eta,\sigma^2)\) we now calculate derivatives for \(S_i(t,\gamma,\eta,\sigma^2)E_{0i}(t,\gamma,\eta,\sigma^2)\) and \(Z_iE_{0i}(t,\gamma,\eta,\sigma^2)\). These are
\begin{align*}
\frac{\partial}{\partial\gamma} S_i(t,\gamma,\eta,\sigma^2)E_{0i}(t,\gamma,\eta,\sigma^2) &= \bigg[S_i(t,\gamma,\sigma^2)^2 + S_i(t,\gamma,\eta,\sigma^2) (S_i(t,\gamma,\sigma^2)-\hat{X}_i(t) \\
&-\gamma\sigma^2\theta_i(t)) +\sigma^2\theta_i(t)dN_i(t)\bigg]E_{0i}(t,\gamma,\eta,\sigma^2) \\
\frac{\partial}{\partial\gamma}Z_iE_{0i}(t,\gamma,\eta,\sigma^2) &= Z_i(2S_i(t,\gamma,\sigma^2)-\hat{X}_i(t)-\gamma\sigma^2\theta_i(t)) E_{0i}(t,\gamma,\eta,\sigma^2)\\
\frac{\partial}{\partial\eta} S_i(t,\gamma,\eta,\sigma^2)E_{0i}(t,\gamma,\eta,\sigma^2) &=  S_i(t,\gamma,\eta,\sigma^2)Z_iE_{0i}(t,\gamma,\eta,\sigma^2) \\
\frac{\partial}{\partial\eta}Z_iE_{0i}(t,\gamma,\eta,\sigma^2) &= Z_i^TZ_iE_{0i}(t,\gamma,\eta,\sigma^2).
\end{align*}

The function \(S^{(0)}_c\) is a scalar and its derivative will be a \((p+1)\times 1\) column vector. The matrices \(S_c^{(1)}\) and \(C^{(1)}\) are both \((p+1)\times 1\) column vectors. For differentiation, entry \((\partial S^{(1)}_c/\partial (\gamma,\eta)^T)_{i,j}\) represents the \(i^{th}\) element of \(S_c^{(1)}\) differentiated with respect to the \(j^{th}\) element of the vector \((\gamma,\eta)^T\). The functions \(S^{(2)}_c,C^{(2)},C^{(3)},V_c^{(1)},V_c^{(2)}\) and \(V_c^{(3)}\) are all \((p+1)\times(p+1)\) matrices. Using the above calculation, we find the following relationships hold
\begin{align*}
\frac{\partial}{\partial(\gamma,\eta)^T}&S^{(0)}_c(k,t,\gamma,\eta,\sigma^2) = \frac{\partial}{\partial(\gamma,\eta)^T}\left[\frac{1}{n}\sum_{i=1}^nE_{0i}(t,\gamma,\eta,\sigma^2)Y_i(k,t)\right] \\ 
 = &\frac{1}{n}\sum_{i=1}^n\left\{ \begin{array}{c} 2S_i(t,\gamma,\sigma^2)-\hat{X}_i(t)-\gamma\sigma^2\theta_i(t) \\ Z_i \end{array} \right\} E_{0i}(t,\gamma,\eta,\sigma^2)Y_i(k,t) \\
 = &\frac{1}{n}\sum_{i=1}^n  \left\{ \begin{array}{c} S_i(t,\gamma,\sigma^2) \\ Z_i \end{array} \right\}  E_{0i}(t,\gamma,\eta,\sigma^2)Y_i(k,t)+ \frac{1}{n}\sum_{i=1}^n  J_i(k,t, \gamma,\sigma^2) E_{0i}(t,\gamma,\eta,\sigma^2)Y_i(k,t) \\
= & S^{(1)}_c(k,t,\gamma,\eta,\sigma^2) + C^{(1)}(k,t,\gamma,\eta,\sigma^2)
\end{align*}
A similar calculation is performed for the differentiation of \(S^{(1)}_c(t,\gamma,\eta,\sigma^2)\). Temporarily removing dependence of all functions on parameters \(k,t,\gamma,\eta\) and \(\sigma^2\), we have
\begin{align*}
\frac{\partial}{\partial(\gamma,\eta)^T}&S^{(1)}_c(k,t,\gamma,\eta,\sigma^2) = \frac{\partial}{\partial(\gamma,\eta)^T}\left[\frac{1}{n}\sum_{i=1}^n  \left\{ \begin{array}{c} S_i(t,\gamma,\sigma^2) \\ Z_i \end{array} \right\}  E_{0i}(t,\gamma,\eta,\sigma^2)Y_i(k,t)\right] \\ 
 = &\frac{1}{n}\sum_{i=1}^n\left\{ \begin{array}{cc} S_i^2 + S_i(S_i-\hat{X}_i -\gamma\sigma^2\theta_i) +\sigma^2\theta_idN_i  & S_iZ_i \\ Z_i(2S_i-\hat{X}_i-\gamma\sigma^2\theta_i) & Z_i^TZ_i\end{array} \right\} E_{0i}Y_i \\
 =& \frac{1}{n}\sum_{i=1}^n\left\{ \begin{array}{c} S_i \\ Z_i \end{array}\right\} \left\{ \begin{array}{c} S_i \\ Z_i \end{array}\right\}^T  E_{0i}Y_i \\
+& \frac{1}{n}\sum_{i=1}^n\left\{ \begin{array}{c} S_i \\ Z_i \end{array}\right\} \left\{ \begin{array}{c} S_i-\hat{X}_i-\gamma\sigma^2\theta_i \\ 0 \end{array}\right\}^T  E_{0i}Y_i +\frac{1}{n}\sum_{i=1}^n
\begin{bmatrix}
\sigma^2\theta_i(u)dN_i& 0 &\cdots & 0 \\
0 & 0 & \cdots& 0 \\
\vdots & \vdots &\ddots& \vdots \\
0 & 0 & \cdots& 0\end{bmatrix}
E_{0i}Y_i \\
= & S^{(2)}_c(k,t,\gamma,\eta,\sigma^2) + C^{(2)}(k,t,\gamma,\eta,\sigma^2) +  C^{(3)}(k,t,\gamma,\eta.\sigma^2) .
\end{align*}
We are now able to differentiate the function \(E_c(k,t,\gamma,\eta,\sigma^2)\). By the quotient rule and the relationships derived above, we have
\begin{align*}
\frac{\partial}{\partial(\gamma,\eta)^T}&E_c(k,t, \gamma,\eta, \sigma^2) = \frac{\partial}{\partial(\gamma,\eta)^T}\left[\frac{S_c^{(1)}(k,t, \gamma, \eta,\sigma^2)}{S_c^{(0)}(k,t, \gamma, \eta,\sigma^2)}\right] \\
 &= \frac{\frac{\partial}{\partial(\gamma,\eta)^T}S_c^{(1)}(k,t, \gamma, \eta,\sigma^2)}{S_c^{(0)}(k,t, \gamma, \eta,\sigma^2)}-\frac{S_c^{(1)}(k,t, \gamma, \eta,\sigma^2)\frac{\partial}{\partial(\gamma,\eta)^T}S_c^{(0)}(k,t, \gamma, \eta,\sigma^2)}{\left[S_c^{(0)}(k,t, \gamma, \eta,\sigma^2)\right]^2} \\
&= \frac{S^{(2)}_c(k,t,\gamma,\eta,\sigma^2)}{S_c^{(0)}(k,t, \gamma, \eta,\sigma^2)} + \frac{C^{(2)}(k,t,\gamma,\eta,\sigma^2)}{S_c^{(0)}(k,t, \gamma, \eta,\sigma^2)} + \frac{C^{(3)}(k,t,\gamma,\eta,\sigma^2)}{S_c^{(0)}(k,t, \gamma, \eta,\sigma^2)} \\
& - \frac{S_c^{(1)}(k,t, \gamma, \eta,\sigma^2)S_c^{(1)}(k,t, \gamma, \eta,\sigma^2)^T}{\left[S_c^{(0)}(k,t, \gamma, \eta,\sigma^2)\right]^2}- \frac{S_c^{(1)}(k,t, \gamma, \eta,\sigma^2)C^{(1)}(k,t, \gamma, \eta,\sigma^2)^T}{\left[S_c^{(0)}(k,t, \gamma, \eta,\sigma^2)\right]^2} \\
&= V_c^{(1)}(k,t, \gamma, \eta,\sigma^2) +  V_c^{(2)}(k,t, \gamma, \eta,\sigma^2) + V_c^{(3)}(k,t, \gamma, \eta,\sigma^2).
\end{align*}

Combining these results, the first derivative of the conditional score function with respect to the vector \((\gamma,\eta)^T\) is given by
\begin{align}
\notag
\frac{\partial}{\partial(\gamma,\eta)^T}U_c(k,\gamma,\eta,\sigma^2) &=\frac{\partial}{\partial(\gamma,\eta)^T}\left[ \int_0^\infty\sum_{i=1}^{n}\left( \{S_i(u, \gamma, \sigma^2), Z_i^T\}^T- E_c(k,u, \gamma,\eta, \sigma^2)\right) dN_i(k,u)\right]\\
\label{eq:d_cond_score}
 = \sum_{i=1}^{n} \int_0^\infty \bigg[\Gamma_i(k,u,&\sigma^2) -V_c^{(1)}(k,u,\gamma,\eta,\sigma^2) +  V_c^{(2)}(k,u,\gamma,\eta,\sigma^2)
+ V_c^{(3)}(k,u,\gamma,\eta,\sigma^2)\bigg]dN_i(k,u).
\end{align}
\section{Details for the proof of Theorem 1}
\begin{lemma}
\label{lemma:taylor}
Starting from
\begin{equation}
\label{eq:est_eqs_all}
\begin{bmatrix}U_c(1, \gamma,\eta,\sigma^2) \\ \vdots \\ U_c(K,\gamma,\eta,\sigma^2)\end{bmatrix} = \begin{bmatrix} \mathbf{0} \\ \vdots \\ \mathbf{0} \end{bmatrix}
\end{equation}
we have that
\[-n^{-\frac{1}{2}}\bar{\mathbf{A}}^{-1}\begin{bmatrix}U_c(1, \gamma,\eta,\sigma^2) \\ \vdots \\ U_c(K, \gamma,\eta,\sigma^2) \end{bmatrix}  = \bar{\mathbf{A}}^{-1}\begin{bmatrix} n^{-1} \frac{\partial}{\partial(\gamma,\eta)^T} U_c(1,\gamma,\eta, \sigma^2)|_{(\gamma^{*(1)}\eta^{*(1)})} \\ \vdots \\n^{-1} \frac{\partial}{\partial(\gamma,\eta)^T} U_c(K,\gamma,\eta, \sigma^2)|_{(\gamma^{*(K)}\eta^{*(K)})}\end{bmatrix}  \cdot n^{\frac{1}{2}}\begin{bmatrix}
\hat{\gamma}_n^{(1)} -\gamma_0 \\
\hat{\eta}_n^{(1)} -\eta_0 \\
\vdots \\
\hat{\gamma}_n^{(K)} -\gamma_0 \\
\hat{\eta}_n^{(K)} -\eta_0 \\
\end{bmatrix}\]
where \((\gamma^{*(k)}\eta^{*(k)})\) lies on the line segment between \((\gamma_0,\eta_0)\) and \((\hat{\gamma}^{(k)}_n,\hat{\eta}^{(k)}_n)\). 
\end{lemma}

\begin{proof}
Fox a fixed, \(k=1,\dots,K\), we begin by applying a Taylor expansion to each element of \(U_c(k,\gamma,\eta,\sigma^2).\) Let \(U_c^j(k,\gamma,\eta,\sigma^2)\) be the \(j^{th}\) element of the \(p\times 1\) column vector \(U_c(k,\gamma,\eta,\sigma^2).\) Let the vector \((\gamma,\eta)^T\) be of dimension \(p\times 1\), then \(\partial/\partial(\gamma,\eta)^T U_c^j(k,\gamma,\eta,\sigma^2)\) is a \(1\times p\) row vector and \((\gamma,\eta)^T-(\gamma_0,\eta_0)^T\) is a \(p\times 1\) column vector. Regularity conditions allow us to perform the following Taylor expansion. For each \(j=1,\dots, p\) the Taylor expansion of \( U_c^j(k,\gamma,\eta,\sigma^2)\) around \((\gamma_0,\eta_0)^T\) is
\[
U_c^j(k,\tilde{\gamma},\tilde{\eta},\sigma^2)=U^j_c(k,\gamma_0,\eta_0,\sigma^2)+ \frac{\partial}{\partial\bs(\gamma,\eta)^T}U_c^j(k,\gamma,\eta,\sigma^2)\bigg\rvert_{(\gamma,\eta)^T=(\gamma^{*(k)}_j,\eta^{*(k)}_j)^T}\cdot((\tilde{\gamma},\tilde{\eta})^T-(\gamma_0,\eta_0)^T)
\]
where \((\gamma^{*(k)}_j,\eta^{*(k)}_j)^T\) lies on the line segment between \((\gamma_0,\eta_0)^T\) and \((\tilde{\gamma},\tilde{\eta})^T\). Each row of \(U_c(k,\gamma_0,\eta_0,\sigma^2)\) represents a dimension, so different rows are differing functions of \((\gamma,\eta)^T\) and hence each \((\gamma^{*(k)}_j,\eta^{*(k)}_j)^T\) is specific to the row \(j.\) 

Given the definition of the estimates \(\hat{\gamma}^{(k)}_n\) and \(\hat{\eta}^{(k)}_n\), we recognise that \(U^j_c(k,\hat{\gamma}^{(k)}_n,\hat{\eta}^{(k)}_n,\sigma^2) = 0.\) Therefore by substituting \((\hat{\gamma}^{(k)}_n, \hat{\eta}^{(k)}_n)^T\) for \((\gamma,\eta)^T\), and stacking the rows, we have
\begin{align*}
\mathbf{0}= \begin{bmatrix}U^1_c(k,\hat{\gamma}^{(k)}_n,\hat{\eta}^{(k)}_n,\sigma^2) \\ \vdots \\U^p_c(k,\hat{\gamma}^{(k)}_n,\hat{\eta}^{(k)}_n,\sigma^2) \end{bmatrix} =  U_c(k,\gamma_0,\eta_0,\sigma^2)+ 
&\begin{bmatrix} \frac{\partial}{\partial(\gamma,\eta)^T}U^1_c(k,\gamma,\eta,\sigma^2)\big\rvert_{(\gamma,\eta)^T=(\gamma^{*(k)}_1,\eta^{*(k)}_1)^T} \\ \vdots \\ \frac{\partial}{\partial(\gamma,\eta)^T}U^p_c(k,\gamma,\eta,\sigma^2)\big\rvert_{(\gamma,\eta)^T=(\gamma^{*(k)}_p,\eta^{*(k)}_p)^T}
\end{bmatrix}\\
&\cdot ((\hat{\gamma}^{(k)}_n,\hat{\eta}^{(k)}_n)^T-(\gamma_0,\eta_0)^T)
\end{align*}
where each \((\gamma^{*(k)}_j,\eta^{*(k)}_j)^T\) lies on the line segment between \((\gamma_0,\eta_0)^T\) and \((\hat{\gamma}^{(k)}_n,\hat{\eta}^{(k)}_n)^T.\)

The matrix \(\mathbf{A}^{(k)}\) by definition is finite valued, symmetric and positive definite. Therefore \(\mathbf{A}^{(k)}\) is invertible with inverse \((\mathbf{A}^{(k)})^{-1}\) which is also finite valued, symmetric and positive definite. Multiplying the above equation by \((\mathbf{A}^{(k)})^{-1}\) and a simple rearrangement gives
\begin{align}
\notag
-n^{-\frac{1}{2}}(\mathbf{A}^{(k)})^{-1}U_c(k,\gamma_0,\eta_0,\sigma^2) = &(\mathbf{A}^{(k)})^{-1}
\begin{bmatrix}n^{-1} \frac{\partial}{\partial(\gamma,\eta)^T}U^1_c(k,\gamma,\eta,\sigma^2)\big\rvert_{(\gamma,\eta)^T=(\gamma^{*(k)}_1,\eta^{*(k)}_1)^T} \\ \vdots \\ n^{-1}\frac{\partial}{\partial(\gamma,\eta)^T}U^p_c(k,\gamma,\eta,\sigma^2)\big\rvert_{(\gamma,\eta)^T=(\gamma^{*(k)}_p,\eta^{*(k)}_p)^T}
\end{bmatrix}\\
\label{eq:taylor_final}
&\cdot n^{\frac{1}{2}}((\hat{\gamma}^{(k)}_n,\hat{\eta}^{(k)}_n)^T-(\gamma_0,\eta_0)^T).
\end{align}

Returning to the definition in the main paper, \(\bar{\mathbf{A}}\) is the block diagonal matrix whose \(k^{th}\) diagonal matrix is the \(p\times p\) matrix \(\mathbf{A}^{(k)}.\) Then, aggregating the above equation for \(k=1,\dots,K,\) we have the result
\begin{equation}
\label{eq:taylor_group}
-n^{-\frac{1}{2}}\bar{\mathbf{A}}^{-1}\begin{bmatrix}U_c(1, \gamma,\eta,\sigma^2) \\ \vdots \\ U_c(K, \gamma,\eta,\sigma^2) \end{bmatrix}  = \bar{\mathbf{A}}^{-1}\begin{bmatrix} n^{-1} \frac{\partial}{\partial(\gamma,\eta)^T} U_c(1,\gamma,\eta, \sigma^2)|_{(\gamma^{*(1)}\eta^{*(1)})} \\ \vdots \\n^{-1} \frac{\partial}{\partial(\gamma,\eta)^T} U_c(K,\gamma,\eta, \sigma^2)|_{(\gamma^{*(K)}\eta^{*(K)})}\end{bmatrix}  \cdot n^{\frac{1}{2}}\begin{bmatrix}
\hat{\gamma}_n^{(1)} -\gamma_0 \\
\hat{\eta}_n^{(1)} -\eta_0 \\
\vdots \\
\hat{\gamma}_n^{(K)} -\gamma_0 \\
\hat{\eta}_n^{(K)} -\eta_0 \\
\end{bmatrix}
\end{equation}
\end{proof}

\begin{lemma}
\label{lemma:A}
\[n^{-1}\frac{\partial}{\partial(\gamma,\eta^T)^T}U_c^{(n)}(k,\gamma,\eta,\sigma_0^2)\bigg\rvert_{\gamma=\gamma^*_n, \eta = \eta^*_n} \xrightarrow{p}\mathbf{A}^{(k)}.\]
\end{lemma}
\begin{proof}
We can write
\[\frac{1}{n}\frac{\partial}{\partial(k,\gamma,\eta^T)^T}U_c^{(n)}(\gamma, \eta, \sigma^2_0)\rvert_{\gamma=\gamma^*, \eta=\eta^*}- \mathbf{A}^{(k)}=-\mathbf{D}_0^{(k)}+\mathbf{D}_1^{(k)}+\mathbf{D}_2^{(k)}+\mathbf{D}_3^{(k)}\]
where
\begin{align*}
\mathbf{D}_0^{(k)}=\frac{1}{n}\sum_{i=1}^{n} &\int_0^\infty\Gamma_i(k,u, \sigma^2_0)dN_i(k,u) -\int_0^\infty \mathbb{E}(\Gamma_i(k,u,\sigma^2_0))s^{(0)}_c(k,u,\gamma_0,\eta_0,\sigma_0^2)h_0(u)du\\
\mathbf{D}_1^{(k)}=\frac{1}{n}\sum_{i=1}^{n} &\int_0^\infty V_c^{(1)}(k,u,\gamma^*,\eta^*,\sigma_0^2)dN_i(k,u) - \int_0^\infty v_c^{(1)}(k,u,\gamma_0,\eta_0,\sigma_0^2)s^{(0)}_c(k,u, \gamma_0,\eta_0,\sigma^2_0)h_0(u)du \\
\mathbf{D}_2^{(k)}= \frac{1}{n}\sum_{i=1}^{n} &\int_0^\infty V_c^{(2)}(k,u,\gamma^*,\eta^*,\sigma_0^2)dN_i(k,u) - \int_0^\infty v_c^{(2)}(k,u,\gamma_0,\eta_0,\sigma_0^2)s^{(0)}_c(k,u, \gamma_0,\eta_0,\sigma^2_0)h_0(u)du \\
\mathbf{D}_3^{(k)}=\frac{1}{n}\sum_{i=1}^{n} &\int_0^\infty V_c^{(3)}(k,u,\gamma^*,\eta^*,\sigma_0^2)dN_i(k,u),
\end{align*}
and it remains to show that each of the terms \(\mathbf{D}_0^{(k)},\mathbf{D}_1^{(k)},\mathbf{D}_2^{(k)}\) and \(\mathbf{D}_3^{(k)}\) converge in probability to zero. For the terms \(\mathbf{D}_0^{(k)}, \mathbf{D}_1^{(k)}\) and \(\mathbf{D}_2^{(k)}\), analogous results are presented by~\citet{andersen1982cox} for survival data. For the term \(\mathbf{D}_3^{(k)}\), we make use of the relationship \(dN_i(k,u)dN_j(k,u)=0\) for \(i\neq j\) and \(dN_i(k,u)^2=1\) since \(dN_i(k,u)\) is an indicator function. The function \(V^{(3)}\) is a \((p+1)\times(p+1)\) matrix where all entries are zero apart from the top left. Therefore we shall only consider entry \((1,1)\). We have that
\begin{align}
\notag
\left[\mathbf{D}_3^{(k)}\right]_{11}=&\frac{1}{n}\sum_{i=1}^{n} \int_0^\infty \left[V^{(3)}(k,u,\gamma,\eta,\sigma^2)\right]_{11} dN_i(k,u) \\[0.1in]
\notag
= &\frac{1}{n}\sum_{i=1}^{n} \int_0^\infty \frac{\left[C^{(3)}(k,u,\gamma,\eta,\sigma^2)\right]_{11}}{\sum_{j=1}^{n}E_{0j}(k,u,\gamma,\eta,\sigma^2)Y_j(k,u)} dN_i(k,u) \\[0.1in]
\notag
 = &\frac{1}{n}\sum_{i=1}^{n} \int_0^\infty \frac{\sum_{j=1}^{n}\sigma^2\theta_j(u)dN_j(k,u)E_{0j}(k,u,\gamma,\eta,\sigma^2)Y_j(k,u)}{\sum_{j=1}^{n}E_{0j}(k,u,\gamma,\eta,\sigma^2)Y_j(k,u)} dN_i(k,u) \\[0.1in]
\label{eq:D3}
 = &\frac{1}{n}\sum_{i=1}^{n} \int_0^\infty \frac{\sigma^2\theta_i(u)E_{0i}(k,u,\gamma,\eta,\sigma^2)Y_i(k,u)}{\sum_{j=1}^{n}E_{0j}(k,u,\gamma,\eta,\sigma^2)Y_j(k,u)} dN_i(k,u). 
\end{align}
A single element in the summand in Expression~\eqref{eq:D3} can be written
\[\int_0^\infty \frac{\frac{1}{n}\sigma^2\theta_i(u)E_{0i}(k,u,\gamma,\eta,\sigma^2)Y_i(k,u)dN_i(k,u)}{n \frac{1}{n} \sum_{j=1}^{n}E_{0j}(k,u,\gamma,\eta,\sigma^2)Y_j(k,u)}\]
and it is clear that \([\mathbf{D}_3^{(k)}]_{11}\xrightarrow{p}0\) as \(n\rightarrow \infty.\) Therefore, we have the result
\[\frac{1}{n}\frac{\partial}{\partial(\gamma,\eta^T)^T}U_c^{(n)}(k,\gamma, \eta, \sigma^2_0)\rvert_{\gamma=\gamma^*, \eta=\eta^*}\xrightarrow{p} \mathbf{A}^{(k)}.\]
\end{proof}

\begin{lemma}
\label{lemma:cov}
\begin{align*}
&Cov\left(W^{(n)}_{j_1}(k_1,\tau_{k_1}, \gamma_0,\eta_0,\sigma_0^2),W^{(n)}_{j_2}(k_2,\tau_{k_2}, \gamma_0,\eta_0,\sigma_0^2)\right) \\
=&\mathbb{E}\left(\sum_{i=1}^n\int_0^{\tau_{k_1}}H_{ij_1}^{(n)}(k_1,u, \gamma_0,\eta_0,\sigma_0^2)H_{ij_2}^{(n)}(k_1,u, \gamma_0,\eta_0,\sigma_0^2)h_0(u)E_{0i}(u,\gamma_0,\eta_0,\sigma_0^2)Y_i(k_1,u)du\right).
\end{align*}
\end{lemma}
\begin{proof}
In the following, we shall drop the dependency of the function \(H^{(n)}\) on the parameters \(\gamma,\eta\) and \(\sigma^2\). In the following calculation, the second equality holds because patients are independent and the orders of summation and integration can be interchanged. Further, the event for each patient can only happen in one analysis so that if \(l_1\neq l_2\) then \(dDN_i(l_1,u)dDN_i(l_2,u)=0\) and because \(dDN_i(l,u)\) is an indicator function, we have that \(dDNi(l,u)^2 = dDN_i(l,u).\) Thus for \(k_1 \leq k_2\), the covariance is given by
\begin{align*}
Cov\bigg(&W^{(n)}_{j_1}(k_1,\tau_{k_1}, \gamma_0,\eta_0,\sigma_0^2),W^{(n)}_{j_2}(k_2,\tau_{k_2}, \gamma_0,\eta_0,\sigma_0^2)\bigg)\\
=&\mathbb{E}\bigg(\int_0^{\tau_{k_1}}\sum_{i=1}^n\sum_{l_1=1}^{k_1}H_{ij_1}^{(n)}(l_1,u)dDN_i(l_1,u)\int_0^{\tau_{k_2}}\sum_{i=1}^n\sum_{l_2=1}^{k_2}H_{ij_2}^{(n)}(l_2,u)dDN_i(l_2,u)\bigg) \\
=&\mathbb{E}\bigg(\sum_{i=1}^n\bigg[\sum_{l_1=1}^{k_1}\int_0^{\tau_{k_1}}H_{ij_1}^{(n)}(l_1,u)dDN_i(l_1,u)\sum_{l_2=1}^{k_2}\int_0^{\tau_{k_2}}H_{ij_2}^{(n)}(l_2,u)dDN_i(l_2,u)\bigg]\bigg) \\
=&\mathbb{E}\bigg(\sum_{i=1}^n\sum_{l=1}^{k_1}\int_0^{\tau_1}H_{ij_1}^{(n)}(l,u)H_{ij_2}^{(n)}(l,u)dDN_i(l,u)\bigg).
\end{align*}

For the following, note that we have
\[\mathbb{E}(dN_i(k,u)|S_i(t,\gamma,\sigma^2),t_i(t),Z_i,Y_i(k,u)=h_0(u)E_{0i}(u,\gamma,\eta,\sigma^2)Y_i(k,u)du\] 
which follows by definition of the conditional intensity process, \(\lambda_i^C(k,t)\). The following calculations use the law of total expectation and the expectation of the counting process \(dN_i(k,u)\) in a similar way to the fixed sample case.
\begin{align*}
\mathbb{E}&\bigg(\sum_{i=1}^n\sum_{l=1}^{k_1}\int_0^{\tau_1}H_{ij_1}^{(n)}(l,u)H_{ij_2}^{(n)}(l,u)dDN_i(l,u)\bigg) \\
 = \mathbb{E}&\bigg(\sum_{i=1}^n\sum_{l=1}^{k_1}\int_0^{\tau_1}\mathbb{E}[H_{ij_1}^{(n)}(l,u)H_{ij_2}^{(n)}(l,u)dDN_i(l,u)|S_i(u, \gamma_0,\sigma_0^2), Z_i, t_i(u),Y_i(l,u)]\bigg) \\
 = \mathbb{E}&\bigg(\sum_{i=1}^n\sum_{l=1}^{k_1}\int_0^{\tau_1}H_{ij_1}^{(n)}(l,u)H_{ij_2}^{(n)}(l,u)\mathbb{E}[dN_i(l,u)-dN_i(l-1, u)|S_i(u, \gamma_0,\sigma^2_0), Z_i, t_i(u),Y_i(l,u)]\bigg) \\
 = \mathbb{E}&\bigg(\sum_{i=1}^n\sum_{l=1}^{k_1}\int_0^{\tau_1}H_{ij_1}^{(n)}(l,u)H_{ij_2}^{(n)}(l,u)h_0(u)E_{0i}(u,\gamma_0,\eta_0,\sigma_0^2)(Y_i(l,u)-Y_i(l-1,u))du\bigg) \\
 = \mathbb{E}&\bigg(\sum_{i=1}^n\int_0^{\tau_{k_1}}H_{ij_1}^{(n)}(k_1,u)H_{ij_2}^{(n)}(k_1,u)h_0(u)E_{0i}(u,\gamma_0,\eta_0,\sigma_0^2)Y_i(k_1,u)du\bigg).
\end{align*}

Therefore we have the result
\begin{align*}
&Cov\left(W^{(n)}_{j_1}(k_1,\tau_{k_1}, \gamma_0,\eta_0,\sigma_0^2),W^{(n)}_{j_2}(k_2,\tau_{k_2}, \gamma_0,\eta_0,\sigma_0^2)\right) \\
=&\mathbb{E}\left(\sum_{i=1}^n\int_0^{\tau_{k_1}}H_{ij_1}^{(n)}(k_1,u, \gamma_0,\eta_0,\sigma_0^2)H_{ij_2}^{(n)}(k_1,u, \gamma_0,\eta_0,\sigma_0^2)h_0(u)E_{0i}(u,\gamma_0,\eta_0,\sigma_0^2)Y_i(k_1,u)du\right).
\end{align*}
\end{proof}

\begin{lemma}
\label{lemma:B}
\[Cov\left(W^{(n)}_{j_1}(k_1,\tau_{k_1}, \gamma_0,\eta_0,\sigma_0^2),W^{(n)}_{j_2}(k_2,\tau_{k_2}, \gamma_0,\eta_0,\sigma_0^2)\right) \xrightarrow{p}B^{(k_1)}_{j_1j_2}.\]
\end{lemma}
\begin{proof}
 The following calculations use the law of total expectation and the expectation of the counting process \(dN_i(k,u)\) in a similar way to the fixed sample case.
\begin{align*}
\mathbb{E}&\bigg(\sum_{i=1}^n\sum_{l=1}^{k_1}\int_0^{\tau_1}H_{ij_1}^{(n)}(l,u)H_{ij_2}^{(n)}(l,u)dDN_i(l,u)\bigg) \\
 = \mathbb{E}&\bigg(\sum_{i=1}^n\sum_{l=1}^{k_1}\int_0^{\tau_1}\mathbb{E}[H_{ij_1}^{(n)}(l,u)H_{ij_2}^{(n)}(l,u)dDN_i(l,u)|S_i(u, \gamma_0,\sigma_0^2), Z_i, t_i(u),Y_i(l,u)]\bigg) \\
 = \mathbb{E}&\bigg(\sum_{i=1}^n\sum_{l=1}^{k_1}\int_0^{\tau_1}H_{ij_1}^{(n)}(l,u)H_{ij_2}^{(n)}(l,u)\mathbb{E}[dN_i(l,u)-dN_i(l-1, u)|S_i(u, \gamma_0,\sigma^2_0), Z_i, t_i(u),Y_i(l,u)]\bigg) \\
 = \mathbb{E}&\bigg(\sum_{i=1}^n\sum_{l=1}^{k_1}\int_0^{\tau_1}H_{ij_1}^{(n)}(l,u)H_{ij_2}^{(n)}(l,u)h_0(u)E_{0i}(u,\gamma_0,\eta_0,\sigma_0^2)(Y_i(l,u)-Y_i(l-1,u))du\bigg) \\
 = \mathbb{E}&\bigg(\sum_{i=1}^n\int_0^{\tau_{k_1}}H_{ij_1}^{(n)}(k_1,u)H_{ij_2}^{(n)}(k_1,u)h_0(u)E_{0i}(u,\gamma_0,\eta_0,\sigma_0^2)Y_i(k_1,u)du\bigg).
\end{align*}
For the following calculation, dependency on parameters \(\gamma_0,\eta_0\) and \(\sigma_0^2\) is removed from the functions \(S^{(0)}_c(k,t), S^{(1)}_c(k,t)\) and \(S^{(2)}_c(k,t)\) for notational purposes. Further, the dependency on parameters \(\gamma_0,\eta_0\) and \(\sigma_0^2\) is also removed from the probabilistic limits \(s^{(0)}_c(k,t),s^{(1)}_c(k,t),s^{(2)}_c(k,t)\) and \(e_c(k,t).\) Then following calculation holds
\begin{align*}
&\sum_{i=1}^n H_{ij_1}^{(n)}(k,u)H_{ij_2}^{(n)}(k,u)h_0(u)E_{0i}(u,\gamma_0,\eta_0,\sigma_0^2)Y_i(k,u)du \\
=&\sum_{i=1}^n n^{-1}\big(\{S_i(u, \gamma_0,\sigma^2_0), Z_i^T\}^T_{j_1}- e_c(k,u)_{j_1}\big)\big(\{S_i(u, \gamma_0,\sigma^2_0), Z_i^T\}^T_{j_2}- e_c(k,u)_{j_2}\big) h_0(u)E_{0i}(u,\gamma_0,\eta_0,\sigma_0^2)Y_i(k,u) \\
=&\bigg([S^{(2)}_c(k,u)]_{j_1j_2}-\frac{[s^{(1)}_c(k,u)]_{j_1}}{s^{(0)}_c(u)}[S^{(1)}_c(k,u)]_{j_2}\\
-&\frac{[s^{(1)}_c(k,u)]_{j_2}}{s^{(0)}_c(k,u)}[S^{(1)}_c(k,u)]_{j_1}+\frac{[s^{(1)}_c(k,u)]_{j_1}}{s^{(0)}_c(k,u)}\frac{[s^{(1)}_c(k,u)]_{j_2}}{s^{(0)}_c(k,u)}S^{(0)}_c(k,u)\bigg)h_0(u) \\ 
\xrightarrow{p}&\bigg(\frac{[s^{(2)}_c(k,u)]_{j_1j_2}}{s^{(0)}_c(k,u)}-\frac{[s^{(1)}_c(k,u)]_{j_1}[s^{(1)}_c(k,u)]_{j_2}}{s^{(0)\otimes 2}_c(k,u)}\bigg)s^{(0)}_c(k,u)h_0(u)\\
=&[v^{(1)}_c(k,u)]_{j_1j_2}s^{(0)}_c(k,u)h_0(u).
\end{align*}

Therefore, combining the above with Lemma~\ref{lemma:cov} it can be seen that for \(1\leq k_1\leq k_2\leq K\),
\begin{align*}
Cov\bigg(&W^{(n)}_{j_1}(k_1,\tau_{k_1}, \gamma_0,\eta_0,\sigma_0^2),W^{(n)}_{j_2}(k_2,\tau_{k_2}, \gamma_0,\eta_0,\sigma_0^2)\bigg) \\
\xrightarrow{p} &\int_0^\infty [v^{(1)}_c(k_1,u, \gamma_0,\eta_0,\sigma^2_0)]_{j_1j_2}s^{(0)}_c(k_1,u, \gamma_0,\eta_0,\sigma^2_0)h_0(u)du \\
\xrightarrow{p} &\left(\int_0^\infty v^{(1)}_c(k_1,u, \gamma_0,\eta_0,\sigma^2_0)s^{(0)}_c(k_1,u, \gamma_0,\eta_0,\sigma^2_0)h_0(u)du\right)_{j_1j_2} \\
=&B^{(k_1)}_{j_1j_2}
\end{align*}
\end{proof}
\section{Sketch proof that the trial is conservative with respect to type 1 error for a group sequential test with \(K=3\)}

We now extend the results of Section~3.2 from the main paper to show that for a group sequential test with \(K=3\) analyses, if the canonical joint distribution does not hold, yet we assume it does anyway, the resulting test will have type 1 error rate less than the planned significance level \(\alpha\).

We follow a similar approach to the proof of the case \(K=2.\)   Let \(\hat{\eta}_1,\hat{\eta}_2,\hat{\eta}_3\) be the sequence of treatment effect estimates that are calculated using the conditional score method. Let the true variance-covariance matrix for this sequence of estimates be \(\Sigma\). Proceeding using method 1, we let the information levels be calculated as \(\mathcal{I}_k = (\Sigma_{kk})^{-1}\) and the \(Z-\)statistic is given by \(Z_k=\hat{\eta}_k\sqrt{\mathcal{I}_k}\) for \(k=1,2,3.\) Under \(H_0\) we have
\begin{equation}
\label{eq:Z_dist}
\begin{bmatrix} \hat{\eta}_1 \\ \hat{\eta}_2 \\\hat{\eta}_3\end{bmatrix} \sim N_K\left[ \left(\begin{array}{c} 0 \\ 0 \\ 0 \end{array}\right),\left( \begin{array}{ccc}
\Sigma_{11} & \Sigma_{12} & \Sigma_{13} \\
\Sigma_{12} & \Sigma_{22} & \Sigma_{23} \\
\Sigma_{13} & \Sigma_{23} & \Sigma_{33}
\end{array} \right) \right], \hspace{1cm}
\begin{bmatrix} Z_1 \\ Z_2 \\ Z_3 \end{bmatrix} \sim N_K\left[ \left(\begin{array}{c} 0 \\ 0 \\ 0\end{array}\right), \left( \begin{array}{ccc}
1 & \rho_{12} & \rho_{13} \\
\rho_{12} & 1 & \rho_{23} \\
\rho_{13} & \rho_{23} & 1 \\
\end{array} \right) \right]
\end{equation}
where the correlation parameters are given by
\begin{equation}
\label{eq:rho}
\rho_{k_1k_2}=Cov(Z_{k_1},Z_{k_2})=\frac{\Sigma_{k_1k_2}}{\sqrt{\Sigma_{k_1k_1}\Sigma_{k_2k_2}}} \: \text{for } 1\leq k_1 < k_2 \leq 3.
\end{equation}

Suppose instead, that we have a different sequence of treatment effect estimates \(\hat{\eta}^*_1,\hat{\eta}^*_2,\hat{\eta}^*_3\) with distribution given below. The values of \(\Sigma_{11},\dots,\Sigma_{33}\) are the same as in~\eqref{eq:Z_dist} and using the same information levels given by \(\mathcal{I}_k = (\Sigma_{kk})^{-1}\), we define the standardised statistics \(Z^*_k=\hat{\eta}^*_k\sqrt{\mathcal{I}_k}\) for \(k=1,2,3.\) The joint distribution of the sequence \(\hat{\eta}^*_1,\hat{\eta}^*_2,\hat{\eta}^*_3\) and the distribution of the \(Z-\)statistics are given by
\begin{equation}
\label{eq:Z_star_dist}
\begin{bmatrix} \hat{\eta}_1^* \\ \hat{\eta}_2^*  \\ \hat{\eta}_3^*\end{bmatrix} \sim N_K\left[ \left(\begin{array}{c} 0 \\ 0 \\0 \end{array}\right),\left( \begin{array}{ccc}
\Sigma_{11} & \Sigma_{22} &  \Sigma_{33} \\
\Sigma_{22} & \Sigma_{22} & \Sigma_{33} \\
\Sigma_{33} & \Sigma_{33} & \Sigma_{33}
\end{array} \right) \right], \hspace{1cm}
\begin{bmatrix} Z_1^* \\ Z_2^* \\ Z_3^* \end{bmatrix} \sim N_K\left[ \left(\begin{array}{c} 0 \\ 0 \\0 \end{array}\right), \left( \begin{array}{ccc}
1 & \rho_{12}^* & \rho_{13}^* \\
\rho_{12}^* & 1 & \rho_{23}^* \\
\rho_{13}^* & \rho_{23}^* & 1
\end{array} \right) \right]
\end{equation}
with correlation parameter
\begin{equation}
\label{eq:rho_star}
\rho_{k_1k_2}^*=Cov(Z_{k_1}^*,Z_{k_1}^*)=\sqrt{\frac{\Sigma_{k_2k_2}}{{\Sigma_{k_1k_1}}}} \: \text{for } 1\leq k_1 < k_2 \leq 3.
\end{equation}

The following conditions, which are an adaptation of the conditions for the case \(K=2\) are assumed to hold.
\begin{conditions}
\label{conditions:efficiency}\
\begin{enumerate}
\item The upper boundary of a group sequential trial, on the \(Z\)-scale, given by \(b_1,\dots, b_3\) is such that \(b_1 \geq b_2\geq b_3\geq 0\)
\item For all \(k_1 < k_2,\) we have \(\Sigma_{k_1k_2}\geq \Sigma_{k_2k_2}\)
\end{enumerate}
\end{conditions}

In a similar manner to the proof for the case \(K=2\) in the main paper, we aim to show that 
\[\mathbb{P}(Z_1>b_1 \cup Z_2 > b_2 \cup Z_3 > b_3) \leq \mathbb{P}(Z_1^* > b_1 \cup Z_2^* > b_3 \cup Z_3^* > b_3).\]
First, note another representation for the above probabilities. For example the probability on the left hand side can be written as
\begin{align}
\notag
&\mathbb{P}(Z_1>b_1) + \mathbb{P}(Z_2>b_2) + \mathbb{P}(Z_3>b_3) \\
\label{eq:intersections}
- & \mathbb{P}(Z_1 > b_1 \cap Z_2 > b_2) - \mathbb{P}(Z_1 > b_1 \cap Z_3 > b_3) - \mathbb{P}(Z_2 > b_2 \cap Z_3 > b_3) \\
\notag
+ & \mathbb{P}(Z_1 > b_1 \cap Z_2 > b_2 \cap Z_3 > b_3)
\end{align}

The marginal distributions are equivalent, that is we have that \(\mathbb{P}(Z_k>b_k)=\mathbb{P}(Z^*_k>b_k)\) for each \(k=1,2,3.\) Let the remaining probabilities of Equation~\eqref{eq:intersections} be summarised by the functions \(f(\cdot)\) and \(g(\cdot)\), which are given by
\begin{align*}
f(\rho_{12},\rho_{13},\rho_{23}) & = f_{12}(\rho_{12}) + f_{13}(\rho_{13}) + f_{23}(\rho_{23}) \\
&=\mathbb{P}(Z_1>b_1\cap Z_2>b_2)+\mathbb{P}(Z_1>b_1\cap Z_3>b_3)+\mathbb{P}(Z_2>b_2\cap Z_3>b_3) \\
g(\rho_{12},\rho_{13},\rho_{23}) &=\mathbb{P}(Z_1>b_1\cap Z_2>b_2\cap Z_3>b_3).
\end{align*}
Then we need to show that 
\begin{equation}
\label{eq:K3}
g(\rho_{12},\rho_{13},\rho_{23})-g(\rho_{12}^*,\rho_{13}^*,\rho_{23}^*)  \leq f(\rho_{12},\rho_{13},\rho_{23})-f(\rho_{12}^*,\rho^*_{13},\rho_{23}^*).
\end{equation}
The right had side of Equation~\eqref{eq:K3} is greater than or equal to zero. This can be seen by a similar argument as for the case \(K=2,\) it is clear that
\begin{align*}
\mathbb{P}(Z^*_1>b_1\cap Z^*_2>b_2) &\leq\mathbb{P}(Z_1>b_1\cap Z_2>b_2) \\
\mathbb{P}(Z^*_1>b_1\cap Z^*_3>b_3) &\leq\mathbb{P}(Z_1>b_1\cap Z_3>b_3) \\
\mathbb{P}(Z^*_2>b_2\cap Z^*_3>b_3) &\leq\mathbb{P}(Z_2>b_2\cap Z_3>b_3)
\end{align*}
and this implies that \(f(\rho^*_{12},\rho^*_{13},\rho^*_{23})\leq f(\rho_{12},\rho_{13},\rho_{23}).\) Suppose that \(g(\rho^*_{12},\rho^*_{13},\rho^*_{23}) > g(\rho_{12},\rho_{13},\rho_{23}),\) then Equation~\eqref{eq:K3} holds so it is sufficient to consider the case \(g(\rho^*_{12},\rho^*_{13},\rho^*_{23}) \leq g(\rho_{12},\rho_{13},\rho_{23}).\) 

Both sides of Equation~\eqref{eq:K3} are greater than or equal to zero, and we know by Conditions~\ref{conditions:efficiency} that \(\rho^*_{12}\leq \rho_{12},\rho^*_{13}\leq \rho_{13}\) and \(\rho^*_{23}\leq \rho_{23}\). For Equation~\eqref{eq:K3} to hold, we need that the function \(f(\cdot)\) is increasing at a greater rate than \(g(\cdot)\) as each parameter \(\rho_{12},\rho_{13}\) and \(\rho_{23}\) increases.

We shall check that Equation~\eqref{eq:K3} holds graphically for each parameter \(\rho_{12},\rho_{13}\) and \(\rho_{23}\) individually. First note that \(f(\rho_{12},\rho_{13},\rho_{23})\) only depends on \(\rho_{12}\) through the function \(f_{12}(\rho_{12})=\mathbb{P}(Z_1>b_1\cap Z_2>b_2)\) and so plotting this probability as a function of \(\rho_{12}\) is sufficient for checking the rate of change of \(f(\rho_{12},\rho_{13},\rho_{23})\) with respect to \(\rho_{12}.\) We cannot check this for every combination of the parameters and so we follow the steps below to check the condition in a systematic way. 
\begin{enumerate}
\item Choose values of \(\rho_{12},\rho_{13},\rho_{23},\rho^*_{12},\rho^*_{13}\) and \(\rho^*_{23}\)
\item Check that \(f(\rho_{12},\rho_{13},\rho_{23})\) increases in \(\rho_{12}\) at a greater rate than \(g(\rho_{12},\rho_{13},\rho_{23})\). Do this by plotting  \(f_{12}(\rho_{12})\) and \(g(\rho_{12},\rho_{13},\rho_{23})\) against \(\rho_{12}\)
\item Check that \(f(\rho_{12}^*,\rho_{13},\rho_{23})\) increases in \(\rho_{13}\) at a greater rate than \(g(\rho_{12}^*,\rho_{13},\rho_{23})\). Do this by plotting  \(f_{13}(\rho_{13})\) and \(g(\rho_{12}^*,\rho_{13},\rho_{23})\) against \(\rho_{13}\)
\item Check that \(f(\rho_{12}^*,\rho_{13}^*,\rho_{23})\) increases in \(\rho_{23}\) at a greater rate than \(g(\rho_{12}^*,\rho_{13}^*,\rho_{23})\). Do this by plotting  \(f_{23}(\rho_{23})\) and \(g(\rho_{12}^*,\rho_{13}^*,\rho_{23})\) against \(\rho_{23}\).
\end{enumerate}

We shall consider three combinations of the parameters \(\rho_{12},\rho_{13},\rho_{23},\rho^*_{12},\rho^*_{13}\) and \(\rho^*_{23}\). The first scenario is based on equally spaced information levels and the \(\rho\) parameters are divided by a factor of 1.2 to get the \(\rho^*\) parameters. The second scenario is where the two interim analyses have higher information levels (than if information was equally spaced) and the \(\rho\) parameters are divided by 1.1. The final scenario is where the two interim analyses have lower information levels (than if information was equally spaced) and the \(\rho\) parameters are divided by 1.3. The parameter values for the three cases are presented in Table~\ref{tbl:rho_choices}.
\begin{table}
\caption{Parameter choices for three cases to compare rate of change of objects \(\mathbb{P}(Z_{k_1}>b_{k_1}\cap Z_{k_2}>b_{k_2})\) and \(\mathbb{P}(Z_{k_1}>b_{k_1}\cap Z_{k_2}>b_{k_2}\cap Z_{k_3} >b_{k_3})\) for a group sequential trial with \(K=3\) analyses .}{%
\begin{tabular}{ccc}
 \\
&\(\rho_{12},\rho_{13},\rho_{23}\)& \(\rho^*_{12},\rho^*_{13},\rho^*_{23}\)  \\
Case 1& \(0.849, 0.693, 0.980\) & \(0.707, 0.577, 0.816\) \\ 
Case 2& \(0.850, 0.765, 0.990\) & \(0.772, 0.695, 0.900\) \\ 
Case 3& \(0.751, 0.531, 0.919\) & \(0.577, 0.408, 0.707\) \\ 
\end{tabular}}
\label{tbl:rho_choices}
\end{table}
\begin{figure}
\centering
\begin{subfigure}[b]{0.3\textwidth}
	\centering
   	\includegraphics[width=1\linewidth]{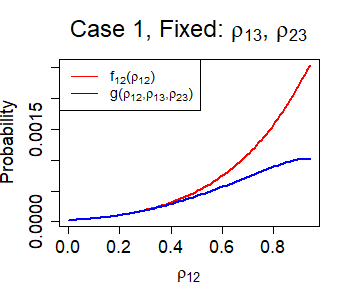}
\end{subfigure}%
\begin{subfigure}[b]{0.3\textwidth}
	\centering
   \includegraphics[width=1\linewidth]{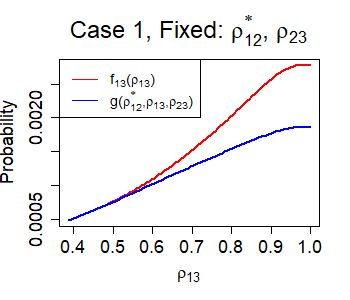}
\end{subfigure}%
\begin{subfigure}[b]{0.3\textwidth}
	\centering
   \includegraphics[width=1\linewidth]{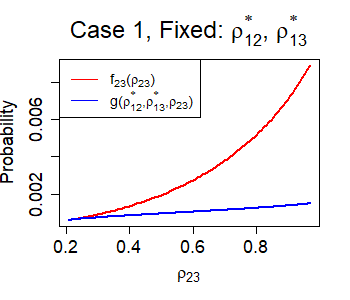}
\end{subfigure}
\begin{subfigure}[b]{0.3\textwidth}
	\centering
   	\includegraphics[width=1\linewidth]{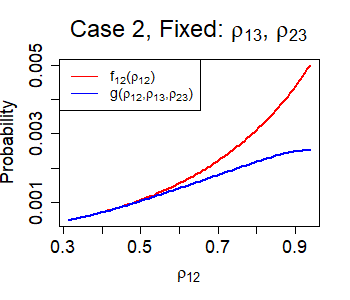}
\end{subfigure}%
\begin{subfigure}[b]{0.3\textwidth}
	\centering
   \includegraphics[width=1\linewidth]{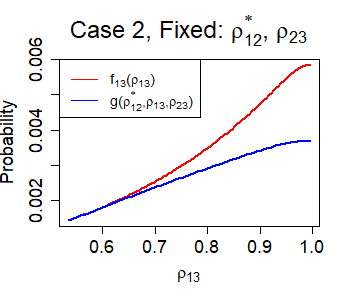}
\end{subfigure}%
\begin{subfigure}[b]{0.3\textwidth}
	\centering
   \includegraphics[width=1\linewidth]{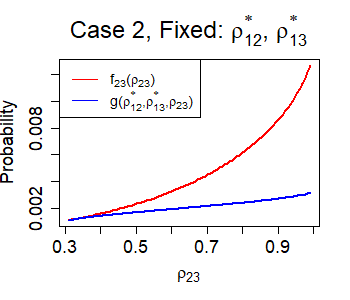}
\end{subfigure}
\begin{subfigure}[b]{0.3\textwidth}
	\centering
   	\includegraphics[width=1\linewidth]{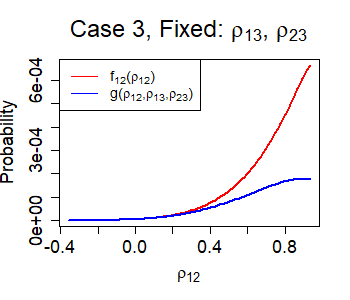}
\end{subfigure}%
\begin{subfigure}[b]{0.3\textwidth}
	\centering
   \includegraphics[width=1\linewidth]{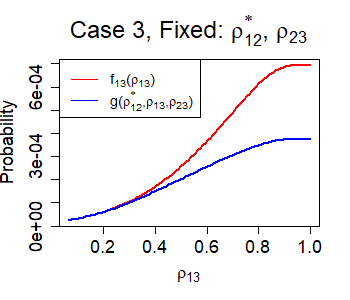}
\end{subfigure}%
\begin{subfigure}[b]{0.3\textwidth}
	\centering
   \includegraphics[width=1\linewidth]{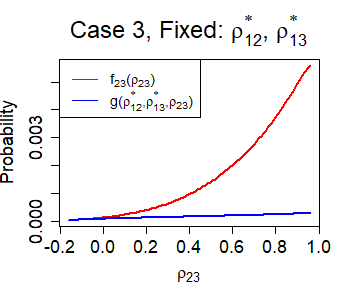}
\end{subfigure}
\caption[Comparison of rate of change of probabilities for \(K=3\).]{Comparison of rate of change of \(f(\rho_{12},\rho_{13},\rho_{23})\) and \(g(\rho_{12},\rho_{13},\rho_{23})\) for a group sequential trial with \(K=3\) analyses. Fixed values of \(\rho_{12},\rho_{13},\rho_{23},\rho_{12}^*,\rho_{13}^*\) and \(\rho_{23}^*\) given by Table~\ref{tbl:rho_choices}.}
\label{fig:P_K3}
\end{figure}
Figure~\ref{fig:P_K3} compares the pairwise probabilities with the probability of crossing all three boundaries for each of the three cases in Table~\ref{tbl:rho_choices}. In cases 1 and 3, as the parameter \(\rho_{13}\rightarrow 1,\) it is not clear whether \(f_{13}(\rho_{13})\) increases at a greater rate than \(g(\rho_{12}^*,\rho_{13},\rho_{23})\). We have checked this numerically and found it to be true. Further the case \(\rho_{13}=1\) is the very rare case where all three analyses have the same information levels. This is not a scenario of concern. Therefore, it is clear that the function \(f(\rho_{12},\rho_{13},\rho_{23})\) increases at a greater rate than the function \(g(\rho_{12},\rho_{13},\rho_{23})\) in each parameter \(\rho_{12},\rho_{13}\) and \(\rho_{23}\) and hence we have shown convincing evidence to prove that Equation~\eqref{eq:K3} holds.

\bibliographystyle{plainnat}  
\bibliography{conditional_score}